\RequirePackage{fix-cm}

\documentclass[preprint,3p,12pt]{elsarticle}
\makeatletter
\def\ps@pprintTitle{%
 \let\@oddhead\@empty
 \let\@evenhead\@empty
 \def\@oddfoot{\centerline{\thepage}}%
 \let\@evenfoot\@oddfoot}
\makeatother

\usepackage{graphicx}
\usepackage{latexsym, amsfonts, amsmath, amssymb, euscript, fancyhdr,mathptmx}
\usepackage{graphicx}
\usepackage{verbatim}
\usepackage{appendix}
\usepackage{mathrsfs}
\usepackage{mathtools}
\usepackage{url}
\usepackage{bbm}
\usepackage{booktabs}
\usepackage{placeins}
\usepackage{float}
\usepackage{framed}
\usepackage{varwidth,array,ragged2e}
\usepackage{rotating}
\usepackage{epstopdf}
\usepackage{pdflscape}
\usepackage{units}
\usepackage{color}
\usepackage{algorithm}
\usepackage[noend]{algpseudocode}
\usepackage{caption}

	\usepackage{amsfonts}
	\usepackage{mathrsfs}
	\usepackage{amsmath}
	\usepackage{amssymb}
	\usepackage{bbm}
	\usepackage{graphicx}
	\usepackage{bbold}
	\usepackage{mathtools}
	\usepackage{amsthm}
	\usepackage{bm}
	\usepackage{xcolor}
\usepackage{tikz}
\newcommand\independent{\protect\mathpalette{\protect\independenT}{\perp}}
\def\independenT#1#2{\mathrel{\rlap{$#1#2$}\mkern2mu{#1#2}}}

\newtheorem{theorem}{Theorem}
\newtheorem{corollary}[theorem]{Corollary}
\newtheorem{proposition}{Proposition}
\newtheorem{lemma}{Lemma}
\newtheorem{assumption}{Assumption}
\newtheorem*{fact}{Fact}

\theoremstyle{definition}

\newcolumntype{L}[1]{>{\raggedright\let\newline\\arraybackslash\hspace{0pt}}m{#1}}
\newcolumntype{C}[1]{>{\centering\let\newline\\arraybackslash\hspace{0pt}}m{#1}}
\newcolumntype{R}[1]{>{\raggedleft\let\newline\\arraybackslash\hspace{0pt}}m{#1}}

\begin{document}

\begin{frontmatter}

\title{Valuation of contingent convertible catastrophe bonds - the case for equity conversion}
\author{Krzysztof Burnecki$^a$, Mario Nicol\'o Giuricich$^b$ and  Zbigniew Palmowski$^a$\footnote{Corresponding author: \texttt{zbigniew.palmowski@pwr.edu.pl}} }

\address{$^a$Faculty of Pure and Applied Mathematics, Hugo Steinhaus Center, Wroclaw University of Science and Technology, Poland\\ $^b$The African Institute for Financial Markets and Risk Management, University of Cape Town, South Africa }

\begin{abstract}
Within the context of the banking-related literature on contingent convertible bonds, we comprehensively formalise the design and features of a relatively new type of insurance-linked security, called a contingent convertible catastrophe bond (CocoCat). We begin with a discussion of its design and compare its relative merits to catastrophe bonds and catastrophe-equity puts. Subsequently, we derive analytical valuation formulae for index-linked CocoCats under the assumption of independence between natural catastrophe and financial markets risks. We model natural catastrophe losses by a time-inhomogeneous compound Poisson process, with the interest-rate process governed by the Longstaff model. By using an exponential change of measure on the loss process, as well as a Girsanov-like transformation to  synthetically remove the correlation between the share and interest-rate processes, we obtain these analytical formulae. Using selected parameter values in line with earlier research, we empirically analyse our valuation formulae for index-linked CocoCats. An analysis of the results reveals that the CocoCat prices are most sensitive to changing interest-rates, conversion fractions and the threshold levels defining the trigger times.
\end{abstract}


\end{frontmatter}

\section{Introduction} \label{sec:introduction}
\noindent Given the pervasiveness of urbanisation in natural catastrophe-prone areas and also the untoward impacts of global warming, insurers, reinsurers and governments have been suffering from substantial natural catastrophe-related losses. Insurers typically deal with this ever-increasing risk by either reinsuring it in the reinsurance market or securitising this risk in the capital markets. Since the capital markets have access to larger, more diversified and more liquid pools of capital as opposed to the equity of reinsurers, such capital markets possess a notable advantage over reinsurance markets when it comes to the financing of catastrophe risk \citep{durbin2001managing}. Therefore, the search for ways of accessing the alternative, rich and robust sources of capital -- from the capital markets -- for the financing of contagion-risk and catastrophe-risk exposed entities has ignited a wave of innovative financial products. In this context, insurance-linked securities (ILS) have been at the fore, with the most prominent type of such products being the catastrophe (CAT) bond, a fully-collateralised debt security which pays off on the occurrence of a pre-defined catastrophic event \citep{cummins2008cat}. Other examples have been catastrophe options, catastrophe futures and catastrophe-equity (CAT-E) puts. However, the market places for each of these instruments are now extinct \citep{braun, wang2016catastrophe}, given low trading volumes. 

On the basis of the firstly the demise of such instruments' marketplaces. secondly the recent expansion in academic literature on these instruments, thirdly the increase in globally-occurring natural catastrophe risk and fourthly the growth of the ILS catastrophe-bond market \citep{swissReport2009}, it seems plausible to suggest the following. There may be a need for novel and alternative sources of funding for catastrophe-prone entities, varying away from the more traditional catastrophe bonds.

Recently in the banking industry, ``contingent capital" instruments, such as contingent convertible (Coco) bonds, have gained the support of various academics, practitioners, economists, regulators and banks as a potential avenue to reduce the need for bailouts of institutions that are classified as `too-big-to-fail' \citep{rudlinger2015contingent, flannery2016stabilizing}. Contingent capital instruments are a type of debt instrument with a loss-absorbing mechanism: that is, they are automatically converted into common equity or written down when a pre-specified trigger event occurs \citep{flannery2016stabilizing}. It is in the very specification of these contingent capital instruments where we see their application to insurance and reinsurance. According to our interviews with industry practitioners and a number of press releases online, many global insurers and also reinsurers such as the Munich Reinsurance Company, SwissReinsurance Company, Hannover Reinsurance Company and the Reinsurance Group of America, are often referred to as being `too-big-to-fail'. Given the success of contingent capital instruments, such as Coco bonds, we now specify a special type of Coco bond for insurers and reinsurers (both of which, in this research, will collectively be referred to as ``issuers"). We believe that the issuance of such a special Coco bond, which we shall call a CocoCat, will help stabilise their issuers' balance sheets in times of distress; in particular in times of extreme natural catastrophes potentially spurring on large non-independent insurance-related losses.

In view of the above, this paper is organised as follows. Section \ref{CocoCats_case} provides a brief discussion of what we believe is the first CocoCat issued, and thereafter goes on to formalise the mechanics, structure and features of such a CocoCat, being cognisant of the literature on traditional Coco bonds. Thereafter, we attempt to give a valuation framework for a specific type of CocoCat - namely one linked to an insurance loss index such as the Property Claims Services (PCS) index. We call this type of CocoCat an Index-linked (IL) CocoCat.

Section \ref{section_setup} describes the necessary joint asset, loss and interest-rate processes under the real-world (or physical) probability measure, needed to price the IL CocoCat in the context of our model. Also, an important assumption is introduced: we assume that natural catastrophe risk and financial markets risk are independent - and such an assumption allows for convenient pricing formulae.

Section \ref{section_valuation} uses the dynamics of the various processes driving the price of the IL CocoCat under the real-world measure to price it under a risk-neutral measure, by introducing a specific measure change suited to the chosen risk-neutral measure. Thereafter, analytical IL CocoCat pricing formulae are derived. In order to simplify pricing formulae and to avoid numerical simulation of the financial markets variables (namely interest-rates and share prices), we employ an exponential change of measure for the loss process and also a Girsanov-like measure change to remove the assumed correlation between interest-rates and share prices.

Section \ref{section_empirics} uses the pricing formulae derived in Section \ref{section_valuation} in order to empirically study the behaviour of the IL CocoCat prices with changes in the model's various parameters. Such an analysis is useful from a contract design perspective for the issuer, for the inclusion of pitch books of ILS structurers, but is also important to the investor in the IL CocoCat (and other types of CocoCats as well). Finally, in Section \ref{section_conclusions}, we state our conclusions and recommendations for further research into this new and interesting topic.

\section{CocoCats - the case for equity convertibles}
\label{CocoCats_case}

\subsection{Background and instrument design}
\label{background_design}
\noindent  In October 2013, a new reinsurance-hybrid security was placed in the capital markets. The Swiss Reinsurance Company (SwissRe) pioneered the creation of a CHF 175 million contingent convertible bond, primarily to sell off hurricane tail risk to a wider pool of investors. The hybrid security has a term of 32 years, pays an annual coupon of 7.5\% and redeems at par, unless triggered. Such returns are reasonable in the ILS markets, wherein ror example EU and USA-based markets investors typically demand returns of 5 to 7\%. The trigger event is unusual in the classical context of Coco (and CAT) bonds, in that it is a dual trigger: the bond triggers if either a 1-in-200 year Atlantic hurricane\footnote{According to press releases by Reuters, Bloomberg and SwissRe, the probability of the occurrence of such an event is low compared to catastrophic events upon which catastrophe bonds are more commonly based (such as 1-in-30 or 1-in-50 year events).} occurs during the term (which is unusual in traditional capital-raising exercises), or SwissRe's solvency ratio (as determined by the Swiss Solvency Test, and reported to the Swiss Financial Market Supervisory Authority at the statutory reporting date) falls below 135\%. Should either trigger event occur, investors lose their entire principal; see \citep{swissReport2013}.

Such novel hybrid securities are, indeed, appealing to both issuers and investors. In a low interest-rate environment and a rising equity market, high-yielding coupon-rates on such novel debt issuances cannot be overlooked by capital markets investors. Moreover, such issuances are attractive to these investors from the perspective of diversification, in that firstly catastrophe risks are remote from financial market risks; and secondly such novel securities differ from the more traditional insurance-linked securities. Nevertheless, it must be borne in mind that the diversification benefits in the case of SwissRe's issue are not as pronounced as in the case of a traditional CAT bond. This is due to the existence of the dual trigger, part of which is based on financial market events. But ultimately, such novel securities -- if structured differently to that of SwissRe's -- can potentially offer rare opportunities to recoup the full amount of principal invested over time should equity prices rise, \textit{ex post} the catastrophe. Finally from the issuer's perspective, such securities could help satisfy regulatory solvency requirements and could reduce probabilities of default \textit{ex post} under Solvency II frameworks in the EU and (in certain cases) the capital-requirement frameworks set out by the National Association of Insurance Commissioners in the USA. Moreover, the coupon payments on such bonds can provide a degree of tax relief for banks \citep{rudlinger2015contingent} and potentially insurers and reinsurers. But most importantly, such instruments can help issuers shift catastrophe-related tail risk off their balance sheets via a novel way and can provide certainty on the capital to be received \textit{ex post} the trigger.

We now nestle the special Coco bond issue by SwissRe into a more formal setting, and attempt to formalise what is meant by a CocoCat. Since the market for such  securities is still in its infancy and given the little scholarly attention to date, it appears clear to posit that, to the extent of our knowledge, no formal definition of a contingent convertible catastrophe (CocoCat) bond exists in the academic literature. In the corporate liabilities sphere, a Coco bond is defined to be a debt instrument\footnote{The debt instrument can be zero-coupon, a fixed-coupon or (more commonly) a floating-coupon (with a fixed spread) bond.} (for accounting purposes, categorised as an ordinary liability%
) that, upon the occurrence of a pre-defined trigger event, converts into common equity via some pre-specified conversion mechanism, or suffers a full write down. \citet{spiegeleer2012pricing} maintain that, in the context of an issuing bank, the trigger event for the Coco is most often a state of possible financial non-viability. Therefore, the purpose of Coco instruments is to stabilise the balance sheet in times of financial distress or contagion effects, and also to allow for a decrease in the systemic risk faced by large financial institutions \citep{rudlinger2015contingent}. However, it must be borne in mind that the conversion to common equity exposes the Coco bond investors to future potential losses and furthermore exposes the existing shareholders to the risk of dilution upon conversion. So, both Coco bond investors and existing shareholders have a greater vested interest in monitoring the risk budget of the financial institution, leading to better risk monitoring by both these parties \citep{de2010countercyclical}. Furthermore, given the impending risk of a dilution, existing shareholders may demand a higher required return on their equity stakes. From the issuer's perspective this higher required return will, to a large extent, by counterbalanced by the reduction in risk from the conversion feature. So, the specification of the conversion feature in a Coco bond's structure -- and also in the context of a CocoCat -- will be important from the perspective of counterbalancing the additional return required by existing shareholders.

We view a CocoCat as a special type of Coco bond. Coco bonds are characterised by two important features - the \textit{conversion trigger} and the \textit{conversion mechanism} - and we attempt to apply them within the sphere of CocoCats. As mentioned by \citet{rudlinger2015contingent}, the conversion trigger sets out one or several events that trigger the conversion mechanism of the Coco bond, while the conversion mechanism explicitly defines, in the bond covenant, what happens to the Coco bond directly after the trigger event. So therefore, we consider a CocoCat to be a Coco bond that has a trigger\footnote{A CocoCat's trigger is allowed to be of the same form as the traditional CAT bond triggers, namely a parametric, index-linked, modelled or indemnity trigger.} linked to the occurrence of a single or sequence of predefined natural catastrophes, and a conversion mechanism whereby the bond either (i) converts into common equity of the issuer (therefore increasing the size of common equity in issue), at a predefined conversion rate as specified in the bond covenant, or (ii) is written down (both principal and coupons) by a fixed percentage which is specified in the covenant. The latter conversion mechanism is reminiscent of the typical structures of various CAT bonds in issue today (for example, see \citet{cummins2009convergence}).

In view of our proposed definition for CocoCats, we offer the following remarks in light of its practical applicability in the insurance and reinsurance settings. Firstly, SwissRe's 2013 placement is loosely an example of such a CocoCat, with the trigger being indemnity-based. Secondly, we note that in the context of bank-issued Cocos with bank-related triggers, in order to protect issuing banks' reputations there is a tendency not to write down the principal amount invested or defer coupon payment \citep{bishop2009contingent}. As mentioned by some industry practitioners we interviewed, this behaviour is also evident in the CAT bond landscape. A number of CAT bonds, when triggered, are not fully written down, but instead begin to pay smaller coupons over a longer period of time (compared to the original term), with the principal repayment potentially being delayed to a time point after maturity. Therefore, by undertaking such actions it is clear that both banks and issuers of CAT bonds do not wish to send negative signals through to current and future investors.  This highlights one of the many potential benefits that CocoCats have to offer to issuers. 

We now present the benefits we foresee CocoCats to offer. Firstly, there can be more certainty in timing (i.e. debt is converted at the time of trigger) of the principal recoveries from the CocoCats since there will be no unexpected delays in principal repayments. Secondly, the structure of the CocoCat can allow for the amount of principal recuperation for the issuer, as well as the total amount to be injected into equity (belonging to the investor) to be fixed in advance. Finally and most importantly, CocoCats can accommodate for a needed reduction in ordinary liabilities and also a provision of immediate \textit{ex ante}\footnote{This \textit{ex ante} capital provision is also a pleasing advantage of a CocoCat over a CAT-E put, since the \textit{ex post} capital provision from the latter may potentially not materialise (i.e. credit risk).} liquidity to immediately pay insurance claims. Notice that the equity conversion feature of the CocoCat can allow for a possible boost in the \textit{ex post} solvency margin of the issuer under the Solvency II regime, upon financial stress caused by the impact of natural catastrophe-related and clustered claims on its balance sheet.

We point out, however, that trigger conversion mechanisms such as write downs can be penal from both the CocoCat issuer's and CocoCat investor's perspective. From the investor's perspective, a high coupon rate prior to the equity conversion will be necessary in light of the write-down risk but must, nonetheless, be commensurate with the risk of trigger of the CocoCat. While from the issuer's perspective, the attractiveness of such a CocoCat will limit the size of the capital market that can be tapped into, as certain investors may not be able or willing to tolerate the risk of a full write down of the principal. On consideration of all of the aforementioned, we propose it may be suitable for CocoCats to be issued on the basis of an equity conversion trigger - in that a certain proportion of the bond's principal is recovered in common equity of the issuer, hence potentially increasing the capital reflected in their balance sheet or in their risk assessment exercises. We subsequently continue with this impetus: from now on, we assume that the CocoCat converts into equity upon trigger.

As final evidence in respect of our case for equity conversion-based CocoCats, we motivate the benefits of issuing CocoCats over traditional CAT bonds. Much of this evidence was pointed out in \citet{georgiopoulos2016valuation}, which we use as a basis for the simple design of our proposed CocoCat. Firstly, CocoCats afford issuers the opportunity to transfer insurance risk without the need to deal or trade their investments in offshore jurisdictions, such as Bermuda and the Cayman Islands, where ILSs are mostly traded. The CocoCat can be directly issued by the issuer via an underwriter or with the help of a structuring agent, in its own local (or judiciously selected foreign) debt market, in a similar fashion to the way a Coco bond is issued. Also, investors can more easily trade in CocoCats compared to CAT bonds, and may not need a qualified investor clause to trade in them \citep{georgiopoulos2016valuation}. Another attractive feature for the issuer of the CocoCat is that the setup of a special purpose vehicle (and also the total return swap required for a CAT bond setup) to issue the debt and act as a type of collateral for the debt is not needed, therefore reducing the instrument's setup expenses. Thirdly because of the non-existence of a special purpose vehicle for their issuance, unlike traditional ILSs, CocoCats do not require a special reinsurance intermediary for the promotion, issue and sale of the debt - an investment bank can underwrite the issue \citep{georgiopoulos2016valuation}. Finally, the trigger mechanisms of CocoCats can easily be based on third-party catastrophe risk models, and specialised in-house model development will, therefore, not be required. We point out that this is also an advantage of issuing an index-linked CAT bond over an indemnity or parametric one, but that basis risk\footnote{Basis risk is the risk of a potential mismatch between the cashflows of the protection instrument and the losses it is supposed to be hedging. We point out that taking on basis risk may be too costly, \textit{ex post}, for the issuer.} can result for both a CocoCat with an index-linked trigger, and an index-linked CAT bond. However, we emphasise that in the case of an insurer or a reinsurer acting as the issuer, the point of a CocoCat is not complete protection against catastrophe-related insured loss payouts, but rather partial financial protection stemming from the capital markets, a markedly larger market compared to the reinsurance market alone.

Therefore, on the basis of the above information, we propose the formal structure for the CocoCat to be as illustrated in Figure \ref{fig_1}. In reference to Figure \ref{fig_1}, the operation of a CocoCat is discussed below.
\begin{enumerate}
\item[2.1.1] The investors transfer the bond's principal to the issuer, where this is then reflected immediately as an ordinary liability in the balance sheet. The proceeds will then be invested by the issuer in liquid treasuries (possibly at a haircut), with discounted mean terms approximately on par with that of the CocoCat in order to avoid credit risk.
\item[2.1.2] The issuer may organise to swap the fixed return for a floating return, especially if the CocoCat's coupons are floating in nature. A floating return is often used so as to base investors' expected returns on a reference interest-rate \citep{muller2000weather}.
\item[2.1.3] Prior to the trigger of the CocoCat, the issuer will use the floating return as well as insurance profits to pay the coupon on the CocoCat, on a pre-specified tenor, to the investors. The coupon will be based on a reference interest-rate (such as LIBOR), and will also include a fixed spread to allow for the catastrophe risk.
\item[2.1.4] If the CocoCat has not been triggered, at maturity the investors will receive their full principal back in cash, together with the final coupon.
\item[2.1.5] If triggered, the CocoCat will terminate and the liability will be written off the issuer's balance sheet. The issuer will  redeem the principal from the treasuries and it will use a predefined proportion of this principal to cover earmarked excess (catastrophic) claims. The remaining proportion of the principal from the treasuries will be converted into (new) common equity, and will belong to the investors in the CocoCat, thereby increasing the issuer's total equity in issue. So, that is, the investors recover a proportion of their principal in equity of the issuer. Figure \ref{fig_2} illustrates the impact of a CocoCat's trigger on the equity and liabilities of the issuer. Notice that after the CocoCat has been triggered, there is a reduction in liabilities (arising from their repayment as a result of the catastrophe), a wipe-out of a proportion of the CocoCat's debt-based value, and a conversion of the remainder of the CocoCat's value into new common equity.
\end{enumerate}

\begin{figure}[h]
\centering
\centering
\begin{tikzpicture}
  [scale=1]
  \tikzstyle{result} = [rectangle,  text width=2.5cm, text centered, draw, fill=gray!30, thick]

  \node[result] (n0) at (-7,5) {{\bf \small{Swap counterparty}}};
  \node[result] (n1) at (-2,5) {{\bf \small{Liquid treasuries}}};
  \node[result] (n2) at (4,-5) {{\bf \small{Investors from capital markets}}};
  \node[result] (n3) at (-2,0) {{\bf \small{Issuer $\qquad$ $\qquad$}}};
  \node[result] (n4) at (4,0) {{\bf \small{Underwriter (i.e. bank)}}};

  \foreach \from/\to in {n0.north east/n1.north west}
    \draw[->, line width=0.5mm] (\from) --(\to)  node[pos=.5, sloped, above] {\footnotesize{Floating return}};
  \foreach \from/\to in {n1.south west/n0.south east}
    \draw[->, line width=0.5mm] (\from) -- (\to)  node[pos=.5,sloped,above] {\footnotesize{Fixed return}};
  \foreach \from/\to in {n1.south/n3.north}
    \draw[->, line width=0.5mm] (\from) -- (\to)  node[pos=.5,sloped,above] {\footnotesize{Floating return}};
  \foreach \from/\to in {n1.south east/n3.north east}
    \draw[->, line width=0.5mm] (\from) -- (\to)  node[pos=.5,sloped,above] {\footnotesize{Principal}};
  \foreach \from/\to in {n3.north west/n1.south west}
    \draw[->, line width=0.5mm] (\from) -- (\to)  node[pos=.5,sloped,above] {\footnotesize{Principal}};
  \foreach \from/\to in {n4.south east/n2.north east}
    \draw[->, line width=0.5mm] (\from) -- (\to)  node[pos=.5,sloped,above] {\footnotesize{Coupons}};
      \foreach \from/\to in {n4.south/n2.north}
    \draw[->, line width=0.5mm] (\from) -- (\to)  node[pos=.5,sloped,above] {\footnotesize{Principal or equity}};
  \foreach \from/\to in {n2.north west/n4.south west/}
    \draw[->, line width=0.5mm] (\from) -- (\to)  node[pos=.5,sloped,above] {\footnotesize{Principal}};
  \foreach \from/\to in {n4.west/n3.east}
    \draw[->, line width=0.5mm] (\from) -- (\to)  node[pos=.5,sloped,above] {\footnotesize{Principal}};
  \foreach \from/\to in {n3.south east/n4.south west}
    \draw[->, line width=0.5mm] (\from) -- (\to)  node[pos=.5,sloped,above] {\footnotesize{Principal or equity}};
      \foreach \from/\to in {n3.north east/n4.north west}
    \draw[->, line width=0.5mm] (\from) -- (\to)  node[pos=.5,sloped,above] {\footnotesize{Coupons}};		
\end{tikzpicture}
\caption{Proposed structure for the CocoCat, in the case of it being underwritten by a bank.} \label{fig_1}
\end{figure}
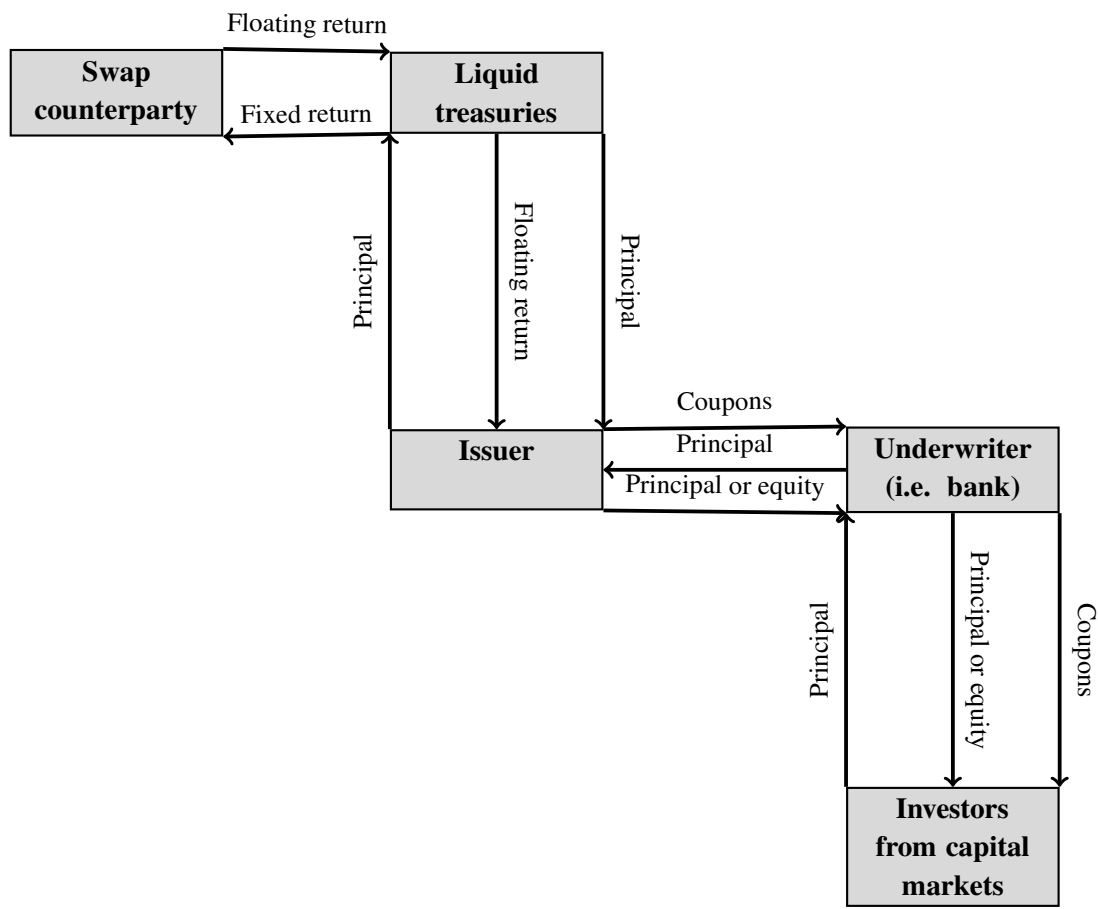

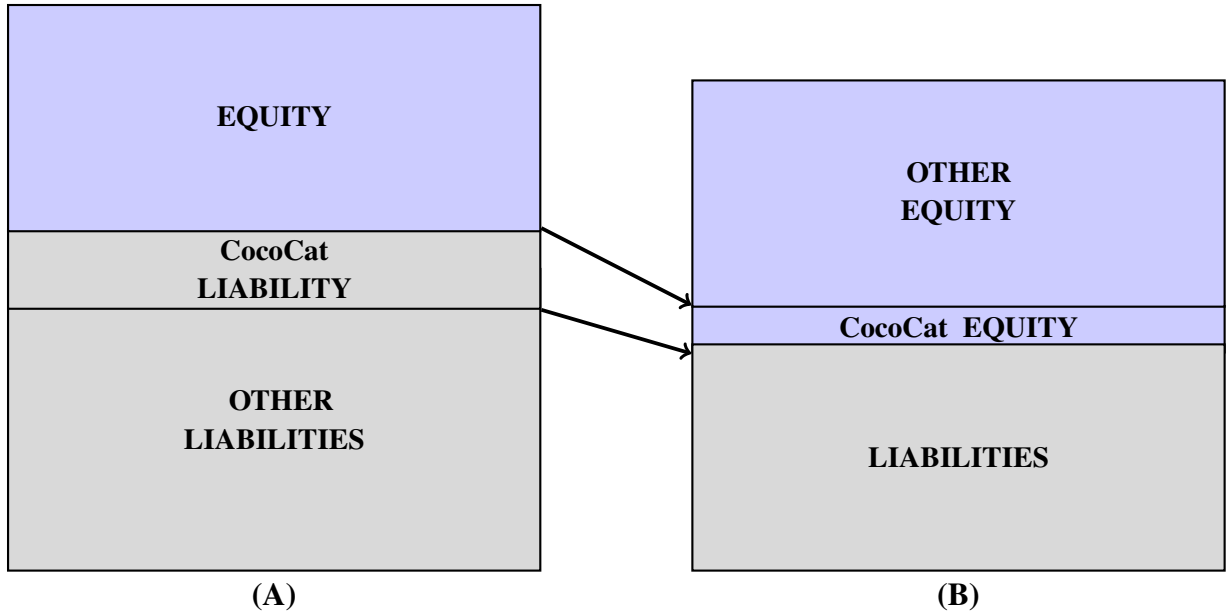
\begin{figure}[h]
\centering
\centering
\begin{tikzpicture}
  [scale=1]
  \tikzstyle{result} = [rectangle,  text width=4cm, text centered, minimum height=4cm,minimum width=7cm, draw, fill=gray!30, thick]
   \tikzstyle{results} = [rectangle,  text width=4cm, text centered, minimum height=3cm, minimum width=7cm, draw, fill=gray!30, thick]
   \tikzstyle{resultss} = [rectangle,  text width=4cm, text centered, minimum height=1cm,minimum width=7cm, draw, fill=gray!30, thick]
   \tikzstyle{resultsss} = [rectangle,  text width=6cm, text centered, minimum height=0.5cm, minimum width=7cm, draw, fill=blue!20, thick]
   \tikzstyle{resultssss} = [rectangle,  text width=6cm, text centered, minimum height=3cm, minimum width=7cm, draw, fill=blue!20, thick]

  \node[result] (n0) [label=below:\textbf{(A)}] at (-6,0) {{\bf \small{ OTHER \\ LIABILITIES}}};
  \node[resultss] (n1) at (-6,2) {{\bf \small{CocoCat \\ LIABILITY}}};
    \node[resultsss] (n2) at (3,1.2) {{\bf \small{CocoCat EQUITY}}};
     \node[resultssss] (n3) at (-6,4) {{\bf \small{EQUITY}}};
  \node[results] (n4) [label=below:\textbf{(B)}] at (3,-0.5) {{\bf \small{LIABILITIES}}};
       \node[resultssss] (n3) at (3,3) {{\bf \small{OTHER \\ EQUITY}}};
\foreach \from/\to in {n1.north east/n2.north west}
    \draw[->, line width=0.5mm] (\from) --(\to)  node[pos=.5, sloped, above] {};
\foreach \from/\to in {n1.south east/n2.south west}
    \draw[->, line width=0.5mm] (\from) --(\to)  node[pos=.5, sloped, above] {};

\end{tikzpicture}
\caption{Projected effect of the CocoCat's trigger on the equity and liabilities of the issuer: (A) provides a simplified overview of the balance sheet structure prior to the trigger of the CocoCat, while (B) provides the balance sheet overview after the trigger of the CocoCat. Notice the decrease in liabilities, as a result of the write-down.} \label{fig_2}
\end{figure}

With a simplified design structure for the CocoCat in mind we now endeavour to analyse and refine the structure more fully, with reference to the ``anatomy" of Coco bonds specified by \citet{spiegeleer2012pricing}. We firstly consider the \textit{conversion trigger}. The natural catastrophe-related trigger event used in the specification of the CocoCat need not be one of the four types of more general Coco bond triggers - accounting, market, regulatory or multi-variate - as put forward by \citet{spiegeleer2012pricing}. This is because of the different purpose and nature of the CocoCat, and also that the CocoCat is a security specific to insurers and reinsurers, and not exclusively banks. The catastrophe-based trigger itself is inherently different from explicit indicators of the financial health of the issuer (such as the solvency margin), but these types of triggers can be indicators of the overall financial health of the insurance and reinsurance industry as a whole. So, we posit that the catastrophe-based trigger is a kind of systemic trigger (which is a trigger linked to the overall financial health of the industry within which the instrument operates, as introduced by \citet{pazarbasioglu2011}).

Before proceeding, it is important to note that the trigger should follow four criteria \citep{rudlinger2015contingent}. We now discuss whether, in general, a natural catastrophe-based trigger is compliant with these four criteria.
\begin{enumerate}
\item[2.2.1] \textbf{Clarity of the trigger event.} The trigger should carry the same message in whatever jurisdiction the issuer operates. It is not possible for any one of the universally-accepted CAT bond triggers (parametric, indemnity, index-linked or modelled) to be a CocoCat trigger. Indeed, indemnity-based triggers would  meet this criterion of clarity with difficulty, given that insurance loss reporting differs from one jurisdiction to the next. Index-linked triggers (if based on a particular insurance loss index) and parametric-based triggers would meet this criterion. The trigger could also encompass an accounting-related trigger, such as in the case of SwissRe's CocoCat, but due care and attention would need to be taken since different accounting regulations apply in different jurisdictions.
\item[2.2.2]  \textbf{Objectivity of the trigger event.} The trigger should be well-documented in the CocoCat's prospectus, and known at the date of issue. There should be categorically no scope to alter the definition of the trigger or change the way that the bond is converted into equity during the term of the CocoCat.
\item[2.2.3]  \textbf{Transparency of the trigger event.} The catastrophe-related trigger should be simple to understand, and observable for both the investor and the issuer at regular intervals of time. This may be a problem for all four of the CAT bond triggers forementioned, since much of the information is proprietary to the company (especially in the case of an indemnity-based trigger) or proprietary to another catastrophe-modelling company (especially in the case of an [industry] index-linked trigger).  This criterion is satisfied in the case of an industry-index trigger. However, we argue that the issuer should endeavour, under strict confidentiality clauses, to provide the information on the evolution of the trigger process to the investors.
\item[2.2.4]  \textbf{Functionality of the trigger event.} Trigger categories for Coco bonds are defined by the functionality condition: the trigger should be an appropriate measure for the state of financial distress of the issuer, or the financial market within which the issuer operates. We argue that the functionality of natural catastrophe-based triggers does give an indication of the financial distress of a particular issuer, since the expected future loss claims and potential claim contagion specific to the issuer, in respect of these catastrophes, will be linked to the occurrence of these catastrophes. Furthermore, the occurrence of catastrophes is not under the control of the issuers, providing further justification for the choice of such objective trigger mechanisms.
\end{enumerate}

We now turn to considering the \textit{conversion mechanism} which, as mentioned before, specifies the procedure to follow at conversion, and the potential loss (of principal) to the investor in the CocoCat at conversion. We follow in the spirit of \citet{rudlinger2015contingent}, and for our proposed CocoCat structure consider the conversion fraction, price and rate.

The conversion fraction, $\zeta$, sets out the proportion of the contingent bond's face value which is converted into common equity, at a contractually-specified conversion price. Thus, if $\zeta = 1$, the full face value of the bond is converted into equity. However, for the purpose of our proposed CocoCat, we suppose that $0 < \zeta < 1$: thus, the CocoCat investor loses a proportion of $1-\zeta$ of his or her invested principal, and has the remaining proportion of $\zeta$ converted into new common equity. Since the main purpose of a CocoCat is to provide immediate, liquid funding in the event of a trigger (which is directly or indirectly linked to the claims experience of the issuer), we propose that the CocoCat debt is written down in the balance sheet, and this written down amount is earmarked to immediately cover any worse-than-expected catastrophe-related claims (directly or indirectly linked to the trigger). This has the effect of delevering the issuer's balance sheet. We do caution, however, that the choice of the conversion fraction is a subjective but critical one. Recall that the purpose of ILSs is not to provide full protection in the case of adverse experience, but rather be complementary to a comprehensive catastrophic risk-management framework. So, the issuer needs to carefully decide on $\zeta$, which ultimately defines how much capital the issuer receives, per unit nominal, \textit{ex post} the catastrophe. Factors which will impact on the magnitude of $\zeta$ are firstly the projected future catastrophe-related claims experience (which is difficult to assess with accuracy) as well as risk budget, and secondly the investor base to which the security will be marketed and issued. Finally, the impact of the consequent equity dilution needs to be accounted for.

The conversion price, $K_P$, is the key indicator for the  potential loss the CocoCat investor can incur on conversion. In fact, $K_P$ can be interpreted as the share price of the CocoCat's issuer at which the fraction, $\zeta$, of the bond's face value is converted into common equity upon trigger. \citet{rudlinger2015contingent} presents three possibilities on how the conversion price can be set, and we apply this to the context of CocoCats: (a) fixed in the CocoCat prospectus, (b) the market share price upon conversion or (c) a function of, \textit{inter alia}, the known share price at time of trigger. We argue that in the case of a CocoCat, option (a) may be unsuitable for investors. Evidence for moderate decreases in insurance firms' share prices after the occurrence of mega natural catastrophes has been found recently by \citet{hagendorff2015impact}, which is intuitively expected. So, there exists the risk of setting the conversion price higher than the market price, creating an adverse effect for the investors of the CocoCat, and consequently investors will receive less share. This could reduce the marketability and relative attractiveness of the CocoCat. Moreover, this conversion price could allow for moral hazard from the side of the issuer and also from current shareholders with substantial stakes in the company. That is, setting the conversion price at an unacceptably high level will  materially affect the equity stake CocoCat investors will recover upon trigger.

However, a key benefit of option (a) is as follows: depending on the size of the CocoCat issue this option may be preferred by current shareholders in the issuing firm because it could potentially restrain dilution of their holdings if the current share price is severely depressed as a result of the impact of natural catastrophes \citep{spiegeleer2012pricing}. In comparison, it appears that option (b) is optimal for CocoCat investors if the full principal amount converts to equity, since hardly any loss on share price differences will result for the investors. Notwithstanding this optimality, (b) is also benficial because CocoCat investors have the potential to gain material stakes in the issuing firm if the size of the CocoCat issue is large relative to the total equity in issue. Also, the scope for the moral hazard identified in option (a) is not explicitly present. Current shareholders, however, will simply have to accept a dilution of their shareholdings \citep{spiegeleer2012pricing}, which may be undesirable from their perspective. Option (c) is also a possibility in CocoCat design and can allow for more flexiblity, which could improve the attractiveness of such an issue to both investor and issuer. There are some further arguments in favour of such function-based conversion prices concerning the reduction in manipulation of share prices by the CocoCat investors - see Section 3.3 of \citet{rudlinger2015contingent}. However, we caution that this avenue may complicate pricing considerations and frameworks. In this research, we accommodate for options (a), (b) and (c), in light of the choice for $K_P$.

We end this section with some comments on the practical use of CocoCats within the context of an insurance or reinsurance company. CocoCats are primarily intended to be of assistance in the management of economic and also solvency capital for an issuer: as postulated by \citet{besson2009much}, capital's critical function is to absorb risks undertaken by the company, be they worse or more contagious than expected. We also believe that our proposed CocoCat lends itself to a situation where existing shareholders may not be called upon that often to provide additional capital in situations of worse-than-expected risk. Requiring additional funds from existing shareholders is unfavourable \citep{besson2009much}.

We also reiterate that the proposed CocoCat instrument is not intended as an ingenious financial instrument to achieve full indemnity against catastrophic losses. Rather, it is to be a complement to and also an integral part of a comprehensive and consistent catastrophe risk-management framework. In consequence, it should adhere to the framework elements put forward by \citet{pazarbasioglu2011}, those being enhanced supervision, a robust economic capital base, transparent disclosure which better informs markets, and a clear resolution regime. Most CocoCats will, for the time-being, be unstandardised, over-the-counter traded and tailor-made (to the issuer) instruments, since their market is new. So careful scrutiny is necessary in developing and managing such issues on behalf of the issuer, the investors and the insurance market regulators. Although it is specific to each issue and is a difficult task, a careful balance between the potential benefits of CocoCats to the issuer, and the rewards reaped by investors, needs to be achieved without the introduction of additional moral hazard and information asymmetries.

\subsection{Comparison to other catastrophe-linked ILSs}
\noindent CocoCats are an ILS that form a unique class of their own. They are similar to CAT-E puts, in that the issuer will sell some of their share to the investor should the trigger occur. However, CAT-E puts suffered from the drawback of credit risk. CocoCats do not, since they provide capital (\textit{ex ante} the trigger) at the outset of the contract. Moreover, CocoCats can be much longer in term than CAT-E puts - the term will depend on the trigger type.

Very much like index-linked catastrophe bonds, catastrophe swaps and CAT-E puts, index-linked CocoCats can expose the issuer to basis risk (which is not the case with industry-loss warranties and reinsurance), especially if the instrument is targeted at hedging a particular portfolio of the issuer's liabilities. But, a pleasing advantage of index-linked ILSs is that they may remove the pervasiveness of moral hazard, which is, unfortunately, an issue when it comes to reinsurance. Finally, it is possible to recover the full principal invested in a CocoCat, even if it is converted (should the equity perform well in the future), but such full recovery is not always possible for CAT bonds. Table \ref{Table_comparison_features} summarises the key points of this comparison.

\begin{table}[H]
  \center
      \caption{Brief comparison of CAT bonds, CAT-E puts and CocoCats (the superscript, *, indicates that the specification can vary beyond what is mentioned in the table).}
  \label{Table_comparison_features}
    \begin{tabular}{l|lll}
    \toprule
          $\quad$ & \textbf{\small{CAT bond}} & \textbf{\small{CAT-E put}} & \textbf{\small{CocoCat}} \\
    \midrule
\footnotesize{\textbf{Term}} & \small{3-5 years} & \small{1-5 years*} & \small{Depends on trigger} \\
\footnotesize{\textbf{Capital provision}} & \small{\textit{Ex-post}} & \small{\textit{Ex-post}} & \small{\textit{Ex-ante}} \\
\footnotesize{\textbf{Possibility of full principal}} & \small{No} & \small{N/A} & \small{Yes} \\
\footnotesize{\textbf{recovery \textit{ex-post} catastrophe}} & $\quad$ & $\quad$ & $\quad$ \\
\footnotesize{\textbf{Trigger of payment}} & \small{Index-linked} & \small{Strike vs. share} & \small{Index-linked} \\
$\quad$ & \small{Indemnity} & $\quad$ & \small{Indemnity} \\
$\quad$ & \small{Pure parametric} & $\quad$ & \small{Pure parametric} \\
$\quad$ & \small{Parametric index} & $\quad$ & \small{Parametric index} \\
$\quad$ & \small{Modelled loss} & $\quad$ & \small{Modelled loss} \\
$\quad$ & \small{Multiple triggers} & $\quad$ & \small{Multiple triggers} \\
\footnotesize{\textbf{Moral hazard}} & \small{Little if index-linked} & \small{None} & \small{None if index-linked} \\
\footnotesize{\textbf{Basis risk}} & \small{Little if pure parametric} & \small{Large - smaller if } & \small{None if pure parametric} \\
$\quad$ & $\quad$ & \small{variance-linked} & $\quad$ \\
\footnotesize{\textbf{Existence of market}} & \small{OTC and exchange} & \small{Extinct} & \small{Very small} \\
\footnotesize{\textbf{Counterparty default risk}} & \small{Low (collateralised)} & \small{High} & \small{Low} \\
\footnotesize{\textbf{Accounting treatment}} & \small{Depends on trigger} & \small{Financial instrument} & \small{Financial instrument} \\

    \bottomrule
    \end{tabular}
\end{table}

\section{Index-linked CocoCat: model} \label{sec:modelSetup}
\label{section_setup}

\subsection{Model setup, assumptions and properties}
\label{subsec:3.1}
\noindent We now turn to focusing on a particular type of CocoCat, and introduce the workings, notations and basic definitions necessary for its analytical pricing. We suppose that the trigger is index-linked, and is furthermore in line with many of the index-linked triggers that CAT bonds are based upon (see, for example, the review by \citet{cummins2008cat}, as well as \citet{haslip2010pricing,mm,nr} and \citet{gatzert2014convergence}). As an example, consider one of the most commonly-issued index-linked CAT bonds: the type that is dependent upon the Property Claims Services (PCS) industry index. The trigger, for most of these bonds, is defined to be the point in time when the accumulated losses from the PCS index exceeed some contractually-specifed threshold level. From here on, we suppose that the CocoCat is based on the PCS loss index, however, note that there is no loss of generality in terms of the type of index which can be used in the model. Other indices (such as that of PERILS in the EU) may also be used in our framework.

We loosely follow \citet{jaimungal2006catastrophe}, and commence under the real-world porbability measure. Under any probability measure the CocoCat's price depends on two emerging phenomena: financial market-related risk and catastrophe-related risk. Since the catastrophe risk will give rise to jumps, we need to work in an incomplete markets setting and moreover note that complicated changes of measure could arise. To avoid this, we make the following assumption in line with much of the previous literature on pricing catastrophe-linked financial instruments. Evidence in support of this assumption has been found by \citet{hoyt} and \citet{cummins2009convergence}, but is disputed by \citet{carayannopoulos2015diversification} and \citet{hagendorff2015impact}.

\begin{assumption} \label{ass_1}
Catastrophe-risk variables and financial markets risk variables are independent in the real-world.
\end{assumption}

This assumption is made in a myriad of research papers on catastrophe-linked instruments, including \citet{taylor}, \citet{cp}, \citet{jarrow2010simple}, \citet{braun}, \citet{mm} and \citet{nr}.
It affords us the possibility to treat catastrophe-risk variables independently from financial markets risk variables. Therefore we can split up the CocoCat pricing into two separate problems under both the real-world and, later, under the risk-neutral probability measure. Moreover the following method is a convenient way to set up the required probability space for the model. We suppose that there exist two probability spaces: for the financial markets risk variables, the space is specified by $(\Omega_F, \hat{\mathcal{F}_\infty}, \mathbb{P}_F)$, where $\hat{\mathcal{F}_\infty} \coloneqq \bigvee_{t \geq 0} \hat{\mathcal{F}_t}$ for the partial financial markets filtration $( \hat{\mathcal{F}_t} )_{t \geq 0}$. Also, $\Omega_F$ is the respective sample space and $\mathbb{P}_F$ is the real-world probability measure for the financial markets risk variables. For the catastrophe-risk variables, $(\Omega_C, \hat{\mathcal{C}_\infty}, \mathbb{P}_C)$, where $\hat{\mathcal{C}_\infty} \coloneqq \bigvee_{t \geq 0} \hat{\mathcal{C}_t}$ for the partial catastrophe risk filtration $( \hat{\mathcal{C}_t} )_{t \geq 0}$. Also, $\Omega_C$ is the respective sample space and $\mathbb{P}_C$ is the real-world probability measure for the catastrophe markets risk variables. From these probability spaces, we can construct a product space $\left(\Omega, \mathcal{G}_\infty, \mathbb{P}  \right)$, where $\Omega \coloneqq \Omega_F \times \Omega_C$, $\mathcal{G}_\infty \coloneqq \hat{\mathcal{F}_\infty} \otimes \hat{\mathcal{C}_\infty}$ and $\mathbb{P} \coloneqq \mathbb{P}_F \otimes \mathbb{P}_C$. Notice how Assumption 1 is conveniently captured in the definition of $\mathbb{P}$. On $\Omega$, we define the following two family of sets which are important for our analyses below: $\mathcal{F}_t \coloneqq \hat{\mathcal{F}}_t \times \{\phi, \Omega_C  \}$ and $\mathcal{C}_t \coloneqq \{\phi, \Omega_F  \} \times \hat{\mathcal{C}}_t $. Also, note that $\mathcal{G}_t \coloneqq \hat{\mathcal{F}}_t \otimes\hat{\mathcal{C}_t} $.

As in the case of the catastrophe swap\footnote{A catastrophe swap is a financial instrument where a protection seller receives periodic payments from a protection buyer and, in exchange, the protection buyer receives a pre-defined loss compensation payment should a pre-agreed trigger event occur.} studied by \citet{braun} and also in the case of CAT bonds and other ILSs, the CocoCat is not an insurance contract but rather a financial instrument, so it is to be priced using financial pricing techniques. As put forward in \citet{cox2004valuation}, if a liquid and large market for catastrophe-linked securities (say CocoCats) exists, then standard derivatives pricing theory (see, for example, \citet{harrison1981martingales}) implies the existence of a risk-neutral measure, so ILSs such as index-linked CocoCats can be priced. However, since index-linked CocoCats rely on a process exhibiting jumps\footnote{Index-linked securities are based on insurance loss indices such as the PCS index. Such indices are often modelled as jump processes, and we return to this point later on in the section.}, the market is incomplete and hence no unique risk-neutral measure exists \citep{embrechts2000actuarial}. Therefore, we assume the existence of a given risk-neutral pricing measure for the financial markets risk variables, $\mathbb{Q}_F$, which has been obtained from suitable calibrations of the interest-rate and stock-price processes.

However, the next question concerns what the associated risk neutral probability measure for the catastrophe-risk variables is. Since the catastrophe-risk variables will be assumed to follow a jump process, we will have several choices \citep{dj}. We consider the incomplete market framework of \citet{merton}. Such an approach has been used extensively in the literature when valuing derivatives with payoffs linked (in some way) to the occurrence of natural catastrophes, see for example \citet{bakshi2002average}, \citet{ly}, \citet{vaugirard2003pricing}, \citet{jaimungal2006catastrophe}, \citet{ly2}, \citet{mm}, \citet{nr} and \citet{chang2017integrated}.  On the grounds of the pervasiveness of Merton's framework, the following assumption is made in our work, and we use it extensively.

\begin{assumption} \label{ass_2}
Investors are risk-neutral towards the jump risk posed by the natural catastrophe-risk variables.
\end{assumption}

More fully and in the context of pricing financial instruments, Assumption \ref{ass_2} states that in the overall economy natural catastrophes can be treated as idiosyncratic risks that can be (almost) fully diversified. The catastrophe risks will pose ``non-systematic risk" and will, in consequence, carry a zero risk-premium. Therefore, the risk-neutral probability measure for the catastrophe-risk variables will coincide with the respective real-world probability measure $\mathbb{P}_C$, and the jump processes will retain their distributional characteristics when changing between measures. For further discussion in support of this, see \citet{delbaen1989martingale}, \citet{cg} and \citet{cp}. However, it must be borne in mind that recent empirical catastrophe bond pricing literature has shown that catastrophe bonds do not have a zero risk premium (see, for example, \citet{ppd}, \citet{braun2} and \citet{gurtler2016impact}); this may carry over to other catastrophe-linked ILS instruments as well. Against this backdrop, it is possible to infer that pricing models based on the zero risk-premium assumption will give rise to values higher than those pricing models which assume a non-zero risk premium. In consequence, the usage of these pricing models may require additional margins added to the calculated value, or margins added to the parameters of the distributions associated with the jump process, all at the discretion of the issuer. Despite this, we remain true to Assumption \ref{ass_2} in our work, for two reasons. Firstly, it is commensurate with actuarial pricing techniques which according to \citet{braun} prevail in practice. Secondly, and most importantly, it can be adapted to the scope of underlying state variables in the model which are not investment assets and hence not tradeable. Hence, we can use real-world data to price. This is useful given the scarcity of (and difficulty of obtaining) pricing data for many catastrophe-linked ILSs, in particular CocoCats.

We are now in a position to construct a risk neutral measure on the measurable product space $\left(\Omega, \mathcal{G}_\infty \right)$: we set $\mathbb{Q} \coloneqq \mathbb{Q}_F \otimes \mathbb{P}_C$. We consider two important random variables defined on the product space $(\Omega, \mathcal{G}_\infty, \mathbb{Q})$. For a financial markets random variable $Y_F: \Omega_F \mapsto \mathbb{R}$ (where $Y_F \in m \hat{\mathcal{F}}_\infty$) we can associate with it a unique $Y \in m{\mathcal{G}}_\infty$, with $Y: \Omega \mapsto \mathbb{R}$, such that 
\begin{align}
Y(\omega_F, \omega_C) \coloneqq Y_F(\omega_F) \label{rv_Y}
\end{align}
for $\omega_F \in \Omega_F$ and $\omega_C \in \Omega_C$. Similarly, for a catastrophe risk random variable $X_C: \Omega_C \mapsto \mathbb{R}$ (where $X_C\in m \hat{\mathcal{C}}_\infty$) we can associate with it a unique $X \in m{\mathcal{G}}_\infty$, with $X: \Omega \mapsto \mathbb{R}$ such that 
\begin{align}
X(\omega_F, \omega_C) \coloneqq X_C(\omega_C). \label{rv_X}
\end{align}
Notice that the definitions of the random variables $X$ and $Y$ provide us with an easy transformation of random variables from the individual spaces to the product spaces. Note, more generally, that $\tilde{X}(\omega_F, \omega_C)$ for some $\tilde{X} \in m\tilde{\mathcal{G}}$ is $\mathcal{C}_\infty-$measurable if and only if it does not depend on $\omega_F$. Similarly, $\tilde{Y}(\omega_F, \omega_C)$ for some $\tilde{Y} \in m\tilde{\mathcal{G}}$ is $\mathcal{F}_\infty-$measurable if and only if it does not depend on $w_C$. 

In the process of constructing the measurable product space from the two probability spaces, we note that Assumption \ref{ass_1} leads to the following proposition. Assumption \ref{ass_1} shows that we can conveniently split the expectation under the risk-neutral probability measure $\mathbb{Q}$.

\begin{proposition} \label{prop_q}
For all the integrable $\mathcal{G}_\infty$-measurable random variables $Y$ and $X$ defined in Equations (\ref{rv_Y}) and (\ref{rv_X}) respectively, it holds that
\begin{align}
\mathbb{E}^{\mathbb{Q}}[XY] &= \mathbb{E}^{\mathbb{Q}_C}[X_C] \mathbb{E}^{\mathbb{Q}_F}[Y_F].
\end{align}
\end{proposition}

\begin{proof}
Suppose that $X_C(\omega_C) \coloneqq \mathbb{I}_A (\omega_C)$ for some $A \in \hat{\mathcal{C}}_\infty$ and $Y_F(\omega_F) \coloneqq \mathbb{I}_B (\omega_F)$ for some $B \in \hat{\mathcal{F}}_\infty$. Therefore by construction, $X(\omega_F, \omega_C)= \mathbb{I}_{ \Omega_F \times A}(\omega_F, \omega_C)$ and $Y(\omega_F, \omega_C)= \mathbb{I}_{B \times \Omega_C }(\omega_F, \omega_C)$. Now,
\begin{align*}
\mathbb{E}^{\mathbb{Q}}[\mathbb{I}_{A \times \Omega_F} \mathbb{I}_{\Omega_C \times B }] &= \mathbb{Q} ( \{A \times \Omega_F\} \cap \{\Omega_C \times B\} ) \\
&= \mathbb{Q}(A \times B) \\
&= \mathbb{Q}_C (A) \mathbb{Q}_F (B) \\
&= \mathbb{E}^{\mathbb{Q}_C} [\mathbb{I}_A] \mathbb{E}^{\mathbb{Q}_F} [\mathbb{I}_B].
\end{align*}
By standard arguments based on the Monotone Class Theorem, the result then holds for all non-negative measurable functions. 
\end{proof}

On $\Omega$, we define the following two sets which are important for our analyses below: $\mathcal{F}_t \coloneqq \hat{\mathcal{F}}_t \times \{\phi, \Omega_C  \}$ and $\mathcal{C}_t \coloneqq \{\phi, \Omega_F  \} \times \hat{\mathcal{C}}_t $. The following proposition gives us the basis for pricing CAT bonds. But before that, let us present the following helpful lemma.

\begin{lemma} \label{lemma_usefulq}
It holds that ${\mathcal{F}_\infty} \independent_{{\mathbb{Q}}}\; {\mathcal{C}_\infty}$
\end{lemma}

\begin{proof}
Suppose that $F \in \mathcal{F}_\infty$ and $C \in \mathcal{C}_\infty$. Then we can write $F = F' \times \Omega_C$ for $F' \in \hat{\mathcal{F}}_\infty$ and $C = \Omega_F \times C'$ for $C' \in \hat{\mathcal{C}}_\infty$. Now
\begin{align*}
\mathbb{Q}(F \cap C) &= \mathbb{E}^{\mathbb{Q}}[\mathbb{I}_F\mathbb{I}_C]\\
&= \mathbb{E}^{\mathbb{Q}}[\mathbb{I}_{F' \times \Omega_C}\mathbb{I}_{\Omega_F \times C'}] \\
&= \mathbb{E}^{\mathbb{Q}_F} [\mathbb{I}_{F'}] \mathbb{E}^{\mathbb{Q}_C} [\mathbb{I}_{C'}]\\
&= \mathbb{Q}_F(F')\mathbb{Q}_C(C')
\end{align*}
where the second last line follows by Proposition \ref{prop_q}. Now, note that
\begin{align*}
\mathbb{Q}_F(F')\mathbb{Q}_C(C') &= \mathbb{Q}_F(F')\mathbb{Q}_C(\Omega_C) \mathbb{Q}_F(\Omega_F) \mathbb{Q}_C(C') \\
&= \mathbb{Q}(F' \times \Omega_C) \mathbb{Q}(\Omega_F \times C').
\end{align*}
\end{proof}

\begin{proposition}\label{prop_split}
Suppose that $\mathcal{G}_\infty$-measurable random variables $Y$ and $X$ are defined as in Equations (\ref{rv_Y}) and (\ref{rv_X}) respectively. Then, it holds that
\begin{align*}
\mathbb{E}^{{\mathbb{Q}}}[{X}{Y} | {\mathcal{G}_t}]&= \mathbb{E}^{{\mathbb{Q}}}[{X}| {\mathcal{F}_t}] \mathbb{E}^{{\mathbb{Q}}}[{Y}| {\mathcal{C}_t}].
\end{align*}
\end{proposition}

\begin{proof}
By the partial averaging property of the conditional expectation,
\begin{align*}
\int_A {X}{Y}\mathrm{d}{\mathbb{Q}} = \int_A \mathbb{E}^{{\mathbb{Q}}}[{X}{Y} | {\mathcal{G}_t}] \mathrm{d}{\mathbb{Q}}
\end{align*}
\noindent for all $A \in \mathcal{G}_t$. We now need to show that 
\begin{align}
\int_{A} \mathbb{E}^{{\mathbb{Q}}}[{X}| {\mathcal{F}_t}] \, \mathbb{E}^{{\mathbb{Q}}}[{Y}| {\mathcal{C}_t}] \, \mathrm{d}{\mathbb{Q}} &= \int_{A}{X}{Y} \mathrm{d}{\mathbb{Q}}. \label{eq:pilambda}
\end{align}
\noindent However, by a Monotone Class argument it is enough to show that Equation (\ref{eq:pilambda}) holds for $A \coloneqq F \times C$, where $F \in {\mathcal{F}_t}$ and $C \in {\mathcal{C}_t} $. Now, Equation (\ref{eq:pilambda}) can be rewritten as
\begin{align*}
 \int_{{\Omega}} \underbrace{\mathbb{E}^{{\mathbb{Q}}}\,[{X}| {\mathcal{F}_t}] \, \mathbb{I}_{F \times \Omega}}_{{\mathcal{F}_t}\text{-measurable}}\,  \underbrace{\mathbb{E}^{{\mathbb{Q}}}[{Y}| {\mathcal{C}_t}]\, \mathbb{I}_{\Omega \times C}}_{{\mathcal{C}_t}\text{-measurable}}\; \mathrm{d} {\mathbb{Q}},
\end{align*}
\noindent which, by the independence of the filtrations ${\mathcal{F}_t}$ and ${\mathcal{C}_t}$ under ${\mathbb{Q}}$ (see Lemma \ref{lemma_usefulq}), is equal to
\begin{align*}
& \quad \left( \int_{{\Omega}} \mathbb{E}^{{\mathbb{Q}}}\,[{X}| {\mathcal{F}_t}] \, \mathbb{I}_{F \times \Omega_C}\; \mathrm{d}{\mathbb{Q}} \right) \left( \int_{{\Omega}} \mathbb{E}^{{\mathbb{Q}}}[{Y}| {\mathcal{C}_t}]\, \mathbb{I}_{\Omega_F \times C} \; \mathrm{d}{\mathbb{Q}} \right) \\
&= \left( \int_{F \times \Omega_C} \mathbb{E}^{{\mathbb{Q}}}[{X} | {\mathcal{F}_t}]\; \mathrm{d}{\mathbb{Q}} \right) \, \left( \int_{\Omega_F \times C} \mathbb{E}^{{\mathbb{Q}}}[ {Y} | {\mathcal{C}_t}] \;\mathrm{d}{\mathbb{Q}} \right) \\
&= \left( \int_{F \times \Omega_C} {X}\; \mathrm{d}{\mathbb{Q}} \right) \, \left( \int_{\Omega_F \times C} {Y}\; \mathrm{d}{\mathbb{Q}} \right) \\
&= \int_{F \times C} {X} {Y} \mathrm{d}{\mathbb{Q}},
\end{align*}
\noindent where the last line follows by the independence of ${X}$ and ${Y}$ under $\mathbb{Q}$.
\end{proof}

Now, before we present our theorem on CAT bond pricing, we provide a helpful corollary to Proposition \ref{prop_split}.

\begin{corollary} \label{prop_smaller}
With the same notation as in Proposition \ref{prop_q}, the following hold:
\begin{enumerate}
\item[(i)] $\mathbb{E}^{\mathbb{Q}}[X | \mathcal{G}_s  ] = \mathbb{E}^{\mathbb{Q}}[X | {\mathcal{F}}_s  ] \quad \forall s > 0; \quad \label{eqn_prop2}$ and
\item[(ii)] $\mathbb{E}^{\mathbb{Q}}[Y | \mathcal{G}_s  ] = \mathbb{E}^{\mathbb{Q}}[Y | {\mathcal{C}}_s  ] \quad \forall s > 0.$
 \end{enumerate} 
\end{corollary}

\begin{proof}
Part (i) follows by setting $Y=1$ in Proposition \ref{prop_q} and by an application of Assumption \ref{ass_2}, while part (ii) follows by setting $X=1$.
\end{proof}

Before presenting the processes capturing the behaviour of the financial markets and catastrophe-risk variables, we introduce some notation. Let:
\begin{itemize}
\item $T>0$ denote the term of the IL CocoCat. For IL CocoCats, we suppose that the term will be in line with that of commonly-issued index-linked CAT bonds. However, for parametric CocoCats based on very rare tail risks (such as the SwissRe CocoCat mentioned in Section \ref{background_design}) it makes sense for the term to be much longer.
\item $Z$ be the principal amount invested in the CocoCat, which the investor will receive back should the CocoCat not trigger during its term.
\item $V_{0}$ denote the price of the IL CocoCat at issue date, $t_0 \coloneqq 0$.
\item $\{S_t, t \geq 0\}$ be the share price of the issuing firm.
\item $\{L_t, t \geq 0\}$ be an aggregate loss process capturing the behaviour of the index upon which the IL CocoCat is based. The aggregate loss process is assumed to have a frequency component specified by a (possibly non-homogenous) Poisson process $N = \{N_t, t \geq 0 \}$ with deterministic intensity specified by the real-valued function $\lambda_t$, and a sequence of iid severity-component continuous random variables $\{ X_k, k \in \mathbb{N} \}$ (independent of the frequency component), each with distribution function $F_X$ and density $f_X$. We assume that $\int_0^T\lambda_s\;\mathrm{d} s <+\infty$.
\item $\mathbb{T} \coloneqq \{t_1, t_2, ... , t_{N-1}, t_N = T \}$ denote the set of $N$ coupon-paying dates.
\item $\Delta$ denote the constant yearly time period between coupon payment dates $t_{i-1}$ and $t_i$ for $i \in \left\{1, 2, ..., n \right\}$. According to \citet{jarrow2010simple}, in the context of CAT bonds this should be either one (1/12 year), three (3/12 year) or six months (6/12 year).
\item $R(t,t_{i-1},t_{i-1} + \Delta)$ be the $\Delta$-year simple forward LIBOR rate per annum, at time $t \geq 0$.
\item $\{r_t, t \geq 0 \}$ be the riskless spot rate process per annum, continuously compounded.
\item $c \geq 0$ be the constant spread for the IL CocoCat (i.e. the catastrophe risk premium).
\item $\zeta$ be the contractually-specified conversion fraction for the IL CocoCat, as introduced before.
\item $D > 0$ be the threshold level for the trigger, specified in the IL CocoCat's prospectus.
\item $\tau = \inf \{0 \leq t \leq T: L_t \geq D \}$, the first time the trigger level is met or exceeded. $D$ is called the contractually-specified threshold level of the IL CocoCat.
\item $K_P$ be the pre-specified conversion price as introduced in Section \ref{background_design}. As mentioned before, $K_P$ can be a pre-specified constant $K$ (we consider this case later on), but it can also be equal to $S_{\tau}$ or $f\left(S_{\tau}  \right)$ for some real-valued function $f$.
\end{itemize}

Based on all of the above notations, our modelling assumptions are encompassed by the following system of stochastic differential equations (SDEs) and identities under the real-world measure $\mathbb{P}$:

\begin{align}
\mathrm{d} r_t &= \bar{\theta}_r (\bar{m}_r - \sqrt{r_t})\mathrm{d}t + \sigma_r \sqrt{r_t}\mathrm{d}W^1_t, \label{eq:1} \\
S_t &= S_t^\mathcal{C}S_t^\mathcal{F}, \label{eq:2} \\
S_t^\mathcal{C} &= \exp \left( -\alpha L_t + \alpha \kappa \int_{0}^t \lambda_u \mathrm{d}u\right) , \label{eq:7} \\
S_t^\mathcal{F} &= S_{0} \exp \left\{Y_t\right\} , \label{eq:8} \\
dY_t &= \mu_S Y_t\mathrm{d}t + \sigma_S Y_t \mathrm{d} W^2_t, \label{eq:3} \\
\mathrm{d} \langle W^1_t, W^2_t  \rangle &= \rho \mathrm{d}t, \label{eq:4} \\
L_t &= \sum_{k=1}^{N_t} X_k. \label{eq:5}
\end{align}

\subsection{Remarks on the model}
\label{section_remarks}
\vspace{3mm}
\noindent \textbf{Interest-rate}\\
\noindent We select the quadratic term structure model of \citet{longstaff1989nonlinear}, in light of some of the shortcomings of the \citet{cir} (CIR) model (see \citet{ahn2002quadratic} and \citet{lo2016pricing} for empirical evidence in favour of quadratic term structure models), but also because it remains the Longstaff process under any Girsanov transformation with a constant kernel (see Theorem \ref{prop_DSR}).
In fact, the entire analysis which follows could be also accomplished for any other interest-rate model as long as the latter property is satisfied. For example, the Vasicek single-factor model of \citet{vasicek} and the Hull-White single-factor model (and consequently the extended Vasicek single-factor model) of \citet{hull1990pricing}\footnote{Also see \citet{hull1993one}.} could be considered.

The Longstaff model takes on the form as shown in Equation (\ref{eq:1}).
It is a two-parameter model wherein the yield is non-linear in $r_t$, while the CIR model has three parameters and a linear yield in the rate $r_t$.
Under $\mathbb{P}$, $\bar{\theta}_r$ and $\bar{m}_r$ are model parameters, $\sigma_r$ is the instantaneous volatility and $W^1_t$ is a standard Brownian motion. Note that $\bar{m}_r = \nicefrac{\sigma_r^2}{4\bar{\theta}_r}$ and $\bar{\theta}_r, \sigma_r > 0$.

\begin{theorem} \label{prop_DSR}
Consider the dynamics of the Longstaff model under any probability measure, $\bar{\mathbb{P}}$:
\begin{align}
\mathrm{d} r_t &= \hat{\theta}_r (\hat{m}_r - \sqrt{r_t})\mathrm{d}t + \hat{\sigma}_r \sqrt{r_t}\mathrm{d}\bar{W}_t, \label{general_DSR}
\end{align}
where $\bar{W}_t$ is a standard Brownian motion under $\bar{\mathbb{P}}$, $\hat{\theta}_r > 0, \hat{\sigma}_r > 0$ and $\hat{m}_r = \nicefrac{(\hat{\sigma}_r)^2}{4\hat{\theta}_r}$. Then if $\bar{\bar{\mathbb{P}}}$ is defined by a Girsanov transform with constant kernel, $\gamma$, i.e. for all $t > 0$,
\begin{align*}
\left.\frac{\mathrm{d}\bar{\bar{\mathbb{P}}}}{\mathrm{d}\bar{\mathbb{P}}}\right\vert_{\hat{\mathcal{F}}_t} &= \hat{\eta}(t),
\end{align*}
for
\begin{align*}
\hat{\eta}(t):= e^{\gamma \bar{W_t} - \frac{1}{2}\gamma^2t},
\end{align*}
then
the dynamics of the interest-rate process under $\bar{\bar{\mathbb{P}}}$ still follows the Longstaff model, that is
\begin{align}
\mathrm{d} r_t &= \tilde{\theta}_r \left(\tilde{m}_r - \sqrt{r_t}\right)\mathrm{d}t + \hat{\sigma}_r \sqrt{r_t}\mathrm{d}\bar{\bar{W}}_t, \nonumber
\end{align}
\noindent where $\bar{\bar{W}}_t \coloneqq \bar{W}_t - \gamma t$ is a standard Brownian motion under $\bar{\bar{\mathbb{P}}}$ and
\begin{align}
\tilde{\theta}_r &= (\hat{\theta}_r - \gamma \hat{\sigma}_r)\nonumber\\
\tilde{m}_r &= \frac{\hat{m}_r \hat{\theta}_r }{\hat{\theta}_r -  \gamma \hat{\sigma}_r}.\nonumber
\end{align}
\noindent Moreover, the property that $\tilde{m}_r = \nicefrac{(\hat{\sigma}_r)^2}{4\tilde{\theta}_r}$ remains.
\end{theorem}

\begin{proof}
\noindent Under $\bar{\bar{\mathbb{P}}}$, the dynamics of the Longstaff model can be expressed as
\begin{align}
\mathrm{d} r_t &= \hat{\theta}_r \left(\hat{m}_r - \sqrt{r_t}\right)\mathrm{d}t + \hat{\sigma}_r \sqrt{r_t}(\mathrm{d}\bar{\bar{W}}_t + \gamma\mathrm{d}t ). \label{eqn:newdsr}
\end{align}
\noindent After a little algebra, Equation (\ref{eqn:newdsr}) can be expressed in the form of the Longstaff interest-rate model:
\begin{align}
\mathrm{d} r_t &= (\hat{\theta}_r - \gamma \hat{\sigma}_r) \left(\frac{\hat{m}_r \hat{\theta}_r }{\hat{\theta}_r -  \gamma \hat{\sigma}_r} - \sqrt{r_t}\right)\mathrm{d}t + \sigma_r \sqrt{r_t}\mathrm{d}{\bar{\bar{{W}}}}_t.  \nonumber
\end{align}
\noindent Moreover, it is easily verified that $\tilde{m}_r = \nicefrac{{(\hat{\sigma}_r)^2}}{4\tilde{\theta}_r}$.
\end{proof}

It must also be noted that the Longstaff model (also called Double Square Root Model) as shown in Equation (\ref{general_DSR}) admits a closed-form solution to the price, $P(r,s, \hat{\theta}_r, \hat{\sigma}_r)$ ($s \geq 0$), of a zero-coupon bond paying $1$-unit of currency at maturity in $s$ years, without an explicit boundary condition at $r_t = 0$ (see \citet{longstaff1989nonlinear}, \citet{beaglehole1992corrections} and \citet{lo2016pricing}).
Indeed, following classical arguments based on the Feynman-Kac formula it is the solution to the partial differential equation specified by:
\begin{align}
\frac{\hat{\sigma}_r}{4}r \frac{\partial^2 P}{\partial r^2}+ \left(\frac{\hat{\sigma}_r}{4} - \hat{\theta}_r \sqrt{r} \right) \frac{\partial P}{\partial r} - r P - \frac{\partial P}{\partial s} &= 0 \label{eqn_pde}
\end{align}
\noindent with initial condition at $s=0$ being one. Solving Equation (\ref{eqn_pde}) by the separation of variables technique yields the bond-pricing function,
\begin{align}
P(r, s,\hat{\theta}_r, \hat{\sigma}_r) &= A^{\text{DSR}}(s) \exp(B^{\text{DSR}}(s)r + C^{\text{DSR}}(s) \sqrt{r}  ), \label{ZCBprice}
\end{align}
\noindent where
\begin{align}
A^{\text{DSR}}(s) &= \left( \frac{2}{1+e^{\psi s}} \right)^{1/2} \exp \left( c_1 + c_2 s + \frac{c_3}{1+e^{\psi s}} \right), \nonumber\\
B^{\text{DSR}}(s) &= \frac{- \psi}{(\hat{\sigma}_r)^2} + \frac{2\psi}{(\hat{\sigma}_r)^2 \left(1+e^{\psi s} \right)},\nonumber\\
C^{\text{DSR}}(s) &= \frac{2\hat{\theta}_r \left(1 - e^{\nicefrac{\psi s}{2}} \right)^2}{ (\hat{\sigma}_r)^2\left(1 + e^{{\psi s}} \right)} \quad \text{and} \nonumber\\
\psi &= \sqrt{2 (\hat{\sigma}_r)^2}; \;
c_1 = \frac{(\hat{\theta}_r)^2}{\psi(\hat{\sigma}_r)^2}\nonumber; \;
c_2 = \frac{\psi}{4} - \frac{(\hat{\theta}_r)^2}{\psi^2}\nonumber; \;
c_3 = \frac{-4 (\hat{\theta}_r)^2 }{\psi^3}\nonumber.
\end{align}

\noindent \textbf{Share price}\\
\noindent We choose the share price process in a similar fashion to, amongst others,  \citet{cox2004valuation},  \citet{jaimungal2006catastrophe}, \citet{lin2009pricing} and \citet{wang2016catastrophe}. Notice that our share price process comprises two components: $S_t^\mathcal{C}$ and $S_t^\mathcal{F}$, the former being the component driven by catastrophe-risk variables and the latter the component driven by financial markets risk variables.  More specifically, when catastrophic events affecting an issuer occur, share prices can be expected to decrease since these large claims must be paid. This is accounted for in the share price process, hence the negative dependency of the process on $L_t$. Under $\mathbb{P}$, $S _{0}$ is the initial share price, $\mu_S$ the long-run mean of $Y_t$ and $\sigma_S$ the instantaneous volatility of $Y_t$. The constant $\alpha > 0$ represents the effect of the catastrophic losses on the logarithm of the share price. The greater the value of $\alpha$, the more serious the effect of the catastrophe losses through the term $\alpha L_t$. Moreover, as in \citet{wang2016catastrophe}, the mean number of claims $\int_0^t \lambda_u \mathrm{d}u$ is included to compensate positively (to some extent) for the presence of downward jumps in the share price. The constant $\kappa > 0$ governs the manifestation of this effect, and is selected on the basis of Lemma 1 below. Note that $W_t^2$ is standard Brownian motion under $\mathbb{P}$.  \\

\noindent \textbf{Aggregate loss}\\
\noindent We follow a classical approach to modelling aggregate loss process, in that we employ a time-inhomogeneous compound Poisson process to govern the behaviour of the IL CocoCat's underlying index. \citet{em} state that such a process is a suitable candidate to model catastrophic losses for catastrophe-related derivatives. We also choose such a process since it can capture over-dispersion in the catastrophe arrivals.\\

\noindent \textbf{Correlations}\\
\noindent We assume that the interest-rate and share price processes are dependent: this is captured by a correlation coefficient of $\rho$.

\subsection{Instrument operation}
\noindent From the time of issue, the IL CocoCat holder will receive (at contractually-specified constant intervals $\Delta$) floating coupon payments based on the $\Delta$-year LIBOR rate, $R(t_{i-1}, t_{i-1}, t_{i-1} + \Delta)$, plus a spread $c$. Hence, the floating payment received at each coupon-date is $R(t_{i-1}, t_{i-1}, t_{i-1} + \Delta) + c$ per unit nominal. At time $T$, the principal of the bond, $Z>0$, is received back by the investor, unless the IL CocoCat is triggered earlier.

During the term of the IL CocoCat, the issuer will monitor the performance of the underlying index, and will record the cumulative losses giving rise to the process $L_t$. If the trigger event is to occur, that is $L_{\tau} > D$ for some $\tau \in [0,T]$ (if it exists), then the IL CAT bond-type leg of the IL CocoCat terminates, and the investor receives a share in the common equity of the firm, at a conversion price equal to $K_P$. That is, the investor will recover $\nicefrac{\zeta Z}{K_P}$ units of shares of the issuer, at a value of $\left(\nicefrac{\zeta Z}{K_P}\right) S_{\tau}$ for $Z$ nominal invested in the IL CocoCat.

\subsection{Selection of conversion price $K_P$}
\label{section_selection}
\noindent As we mentioned before, there are three general cases one can consider for the conversion price, and we consider each of them. For the case when $K_P$ is assumed to be a real-valued function of the share price, we suppose that it takes on the form form $K_P \coloneqq S_\tau^\nu$, for $\nu \in (0,1]$. Taking such a functional of $S_\tau$ allows for flexibility in the design of the IL CocoCat, and also analytical expressions for the price of the IL CocoCat. Firstly, it allows both the investor and the issuer to take into account their views on the impact of catastrophe-related losses on the issuer's share price. For example, if the investor believes that the market does not satisfactorily capture the impact of large catastrophic losses on its assessment of the share's value, then $\nu$ should be set equal to a value less than $1$ so that the investor purchases the share at a value cheaper than market value. Secondly, it can allow both the investor and the issuer to account for their views on the future market performance of the issuer's share. For example, if the issuer believes the share price will rise (independent of the catastrophic losses), then $\nu$ could be set equal to a value less than $1$, so that the investor does not gain a large portion of share ownership. Finally, note that if $\nu = 1$, then the conversion price is set at the share price at time-of-conversion.

\section{Analytical  risk-neutral pricing of index-linked CocoCat}
\label{section_valuation}

\noindent We now price the index-linked CocoCat within the context of our model, commencing with a generic conversion price $K_P$, at the issue date $t_0 \coloneqq 0$ under our risk-neutral probability measure $\mathbb{Q}$.
All notation is assumed to follow that in Section \ref{subsec:3.1}, and the assumptions as well as results therein are also assumed to hold. As specified in Section \ref{subsec:3.1},  \begin{equation}\label{R}R(t,t_{i-1},t_{i-1} + \Delta)\end{equation} is the forward LIBOR at time $t$ for the interval $[t_{i-1},t_{i-1}+ \Delta]$. In particular, $R(t,t,t + \Delta)$ is the LIBOR at time $t$: since $\Delta$ is assumed to be constant, we denote the LIBOR process by $\{R_t, t \geq 0 \}$.
\begin{equation}\label{B}
B(0,t) \coloneqq \exp\left(-\int_0^t r_u \mathrm{d}u  \right)\end{equation}
 is the discounted riskless bank account associated with the progressively measurable process $\{r_t, t \geq 0 \}$, and \begin{equation}\label{P}P(0,t) \coloneqq \mathbb{E}^{\mathbb{Q}} \left[ B(0,T) | \mathcal{F}_t \right],\qquad \forall t \in [0,T].\end{equation}

Now, under $\mathbb{Q}$, it is possible to find an expression for the price of the CocoCat at issue under expectation. This is so important that we formalise it as a main Fact.

\begin{framed}
\begin{fact}
The issue-date, (or time-zero) risk-neutral price, $V_0$, of an IL CocoCat is
\begin{align}
V_0 &= \mathbb{E}^\mathbb{Q}[ I_1 + I_2 + I_3 ], \label{eq:BIG}
\end{align}
where
\begin{align*}
I_1 &:=  \sum_{i=1}^N  \big(R_{t_{i-1}} + c    \big) \Delta Z\mathbb{I}_{\{\tau > t_i\}}B(0, t_i); \\
I_2 &:= \frac{\zeta Z}{K_P} S_{\tau}\mathbb{I}_{\{\tau \leq T\}}B(0, \tau);\\
I_3 &:= Z\mathbb{I}_{\{\tau > T\}} B(0, T).
\end{align*}

\end{fact}
\end{framed}

Notice that the expectation in the Fact comprises three terms. $I_1$ represents the coupon payments (linked to LIBOR) inclusive of the spread, while $I_3$ represents the capital repaid at maturity should no default occur prior. $I_2$ represents the recovery upon conversion-to-equity. Notice that if $K_P$ is set to be equal to $S_{\tau}$, we consequently obtain the pricing formula for a specific type of CAT bond which pays out $\zeta Z$ immediately upon the time of trigger (i.e. time $\tau$). We emphasise this feature of our model, which lends it to a broader suite of applications. Some CAT bond issues are in practice of this nature so our valuation framework may find potential applicability in this instance. However, in our framework below we shall consider a specific function of $S_\tau$, which shall lead to three cases for the conversion price.

\noindent Now we show how our model changes under the martingale measure $\mathbb{Q}$. Under $\mathbb{Q}$, the discounted share price process $\left\{\exp\left( -\int_{0}^t r_u \mathrm{d}u \right) S_t, t \geq 0 \right\}$ must, by definition of a risk-neutral measure, be a martingale with respect to the filtration $\mathcal{G}_t$. This is gives the following theorem.

\begin{theorem} \label{prop:1}
Let
\begin{equation}\label{kappa}
\kappa = \frac{1}{\alpha}\left(1 - (\mathcal{L}f_X)(\alpha) \right),
\end{equation}
where $(\mathcal{L}f_X)(\alpha):=\int_0^\infty e^{-\alpha y}f_X(y)\mathrm{d}y$ is the Laplace transform  of $f_X$, the density function of each severity component $X$ over the positive support of $X$ evaluated at $\alpha$.
Then there exists the risk-neutral measure $\mathbb{Q}=\mathbb{Q}_F\otimes\mathbb{P}_C$ and the catastrophe-risk and financial markets risk variables under this measure
are captured by the following system of equations:
\begin{align}
\mathrm{d} r_t &= \theta_r (m_r - \sqrt{r_t})\mathrm{d}t + \sigma_r \sqrt{r_t}\mathrm{d}\tilde{W}^1_t, \label{eq:9} \\
S_t &= S_t^\mathcal{C} S_t^\mathcal{F}, \label{eq:10} \\
S_t^\mathcal{C} &= \exp \left( -\alpha L_t + \alpha \kappa \int_{0}^t \lambda_u \mathrm{d}u \right) , \label{eq:15} \\
{S}_t^\mathcal{F} &= S_{0} \exp \left\{ {Y}_t\right\} , \label{eq:16} \\
\mathrm{d} {Y}_t &= r_t Y_t\mathrm{d}t + \sigma_S Y_t\mathrm{d}\tilde{W}^2_t, \label{eq:11} \\
\mathrm{d} \langle \tilde{W}^1_t, \tilde{W}^2_t  \rangle &= \rho \mathrm{d}t,\label{eq:12} \\
L_t &= \sum_{k=1}^{N_t} X_k, \label{eq:13}
\end{align}
 \noindent where $\theta_r$ and $m_r$ are the risk-neutral parameters for the interest-rate process
 given explicitly in Equations \eqref{new1} and \eqref{new2} and
 $\tilde{W}^1_t$ and $\tilde{W}^2_t$ are two Brownian motions under the measure $\mathbb{Q}_F$.
\end{theorem}

\begin{proof}
Equations (\ref{eq:15}) and (\ref{eq:13}) are both an immediate consequence of Assumption 2, in that the processes retain their distributional forms as well as parameters when moving from  $\mathbb{P}$ to $\mathbb{Q}$. Now, we consider how to find Equations (\ref{eq:9}), (\ref{eq:16}), (\ref{eq:11}) and (\ref{eq:12}) from Equations (\ref{eq:1}), (\ref{eq:8}), (\ref{eq:3}) and (\ref{eq:4}) respectively.

In the first step we prove that the discounted stock price
$\left\{\exp\left( -\int_{0}^t r_u \mathrm{d}u \right) S_t^{\mathcal{F}}, t \geq 0 \right\}$
is a $\hat{\mathcal{F}_t}$-martingale under the appropriate chosen market martingale-measure, $\mathbb{Q}_F$.
Classical arguments based on It\^{o}'s formula show that this requirement is equivalent to
Equation \eqref{eq:11}. We show now how to choose the measure $\mathbb{Q}_F$ to obtain this equation out 
of Equation \eqref{eq:3}. Define
\begin{align}
B_t^1 &= W_t^1 \quad \text{and} \nonumber \\
B_t^2 &= \frac{1}{\sqrt{1-\rho^2}}W_t^2 - \frac{\rho}{\sqrt{1-\rho^2}}W_t^1. \nonumber
\end{align}
By L\'evy's Theorem (see \citet[Chapter 3]{karatzas2012brownian}), $B_t^1$ and $B_t^2$ are two Brownian motions respectively. Moreover, the covariance between $B_t^1$ and $B_t^2$ is zero, so the two Brownian motions are uncorrelated.
Let
\begin{align}
\gamma^1&:=\rho\sigma_S
\quad \text{and} \label{mpr_1}\\
\gamma_u^2 &:= \sigma_S\sqrt{1-\rho^2} +\beta_u.
\label{mpr_2}
\end{align}
for some $\beta_u$ which will be specified later.

We define the risk neutral measure $\mathbb{Q}_F$ on the financial market using inverse of Girsanov martingale, which is also a martingale:
\begin{align}
\left.\frac{\mathrm{d}\mathbb{Q}_F}{\mathrm{d}\mathbb{P}_F}\right\vert_{\hat{\mathcal{F}}_t}=\eta(t) \coloneqq \exp \left(-\frac{1}{2} \int_0^t [(\gamma^1)^2 + (\gamma_u^2)^2] \mathrm{d}u - \int_0^t \gamma^1 \mathrm{d}B_u^1 - \int_0^t \gamma_u^2 \mathrm{d}B_u^2\right). \nonumber
\end{align}

Now by the multidimensional Girsanov Theorem (see \citet[Chapter 3]{karatzas2012brownian})
the processes
\begin{align}
\tilde{B}_t^1 &:= B_t^1 + \int_0^t \gamma^1 \mathrm{d}u, \quad \text{and} \nonumber \\
\tilde{B}_t^2 &:= B_t^2 + \int_0^t \gamma_u^2 \mathrm{d}u, \nonumber
\end{align}
are standard (uncorrelated) Brownian motions under the measure $\mathbb{Q}_F$. Then, by a further application of L\'evy's Theorem, we can define two new correlated Brownian motions, under the measure $\mathbb{Q}_F$, such that
\begin{align}
\tilde{W}_t^1 &:= \tilde{B}_t^1=W_t^1+\int_0^t \gamma^1 \mathrm{d}u, \label{pierwszy}\\
\tilde{W}_t^2 &:= \rho \tilde{B}_t^1 + \sqrt{1 - \rho^2}\tilde{B}_t^2=W_t^2+ \int_0^t \left(\rho \gamma^1 + \sqrt{1 - \rho^2}\gamma_u^2\right) \mathrm{d}u,
 \label{drugi}\\
\mathrm{d} \langle \tilde{W}^1_t, \tilde{W}^2_t  \rangle &= \rho \mathrm{d}t. \nonumber
\end{align}
Now by choosing
\begin{equation}\label{beta}
\beta_u = \frac{\mu_S-r_u-\sigma_S^2}{\sigma_S\sqrt{1-\rho^2}}
\end{equation}
we obtain $\mu_s- \sigma_S\left(\rho \gamma^1 + \sqrt{1 - \rho^2}\gamma_u^2\right)=r_u$.
So, inserting this into Equation \eqref{eq:3} and using \eqref{drugi}, we obtain Equation \eqref{eq:11}.

From Equation \eqref{pierwszy} and Theorem \ref{prop:1} it follows that the Longstaff interest-rate model is preserved, that is, Equation \eqref{eq:9} holds true.
In this case, the new parameters are given by
\begin{align}
\theta_r &:=\bar{\theta}_r+\sigma_r\gamma^1, \quad \text{and}\label{new1}\\
m_r &:= \frac{\sigma^2_r}{4 \theta_r}.\label{new2}
\end{align}

\noindent In the second step, we prove that the chosen $\kappa$ satisfies $\mathbb{E}^{\mathbb{P}_C}\left[ S_t^\mathcal{C} | \hat{\mathcal{C}}_s  \right] = S_s^\mathcal{C}$
for $s < t$.
In the context of Assumption \ref{ass_2}, we hence require that, for $s < t$,
\begin{align}
\mathbb{E}^{\mathbb{P}_C} \left[ \exp \left( -\alpha L_t + \alpha \kappa \int_{0}^t \lambda_u \mathrm{d}u \right) | \hat{\mathcal{C}}_s  \right] &= \exp \left( -\alpha L_s + \alpha \kappa \int_{0}^s \lambda_u \mathrm{d}u \right), \label{poisson_first}
\end{align}
which can be rewritten as
\begin{align}
\mathbb{E}^{\mathbb{P}_C} \left[ \exp \left( -\alpha (L_t - L_s) + \alpha \kappa \int_{s}^t \lambda_u \mathrm{d}u \right) | \hat{\mathcal{C}}_s  \right] &= \mathbb{E}^{\mathbb{P}_C} \left[ \exp \left( -\alpha (L_t - L_s) + \alpha \kappa \int_{s}^t \lambda_u \mathrm{d}u \right)\right]=1.  \label{choicekappa}
\end{align}
Now consider $\mathbb{E}^{\mathbb{P}_C} \left[ \exp \left( -\alpha (L_t - L_s) \right) \right]$. This can be simplified as follows.
\begin{align}
\mathbb{E}^{\mathbb{P}_C} \left[ \exp \left( -\alpha (L_t - L_s) \right) \right] &= \mathbb{E}^{\mathbb{P}_C} \left[ \mathbb{E}^{\mathbb{P}_C} \left[ \exp \left( -\alpha (L_t - L_s)\right) | N_t - N_s  \right] \right] \nonumber \\
 &=  \mathbb{E}^{\mathbb{P}_C} \left[ \exp \left( -\alpha \sum_{k=1}^{N_t - N_s} X_k \right) | N_t - N_s  \right]\nonumber\\
 &= \mathbb{E}^{\mathbb{P}_C} \left[ \{(\mathcal{L}f_X)(\alpha)\}^{N_t - N_s}  \right] \nonumber\\
 &= G_{N_t - N_s} \left( (\mathcal{L}f_X)(\alpha) \right) \nonumber\\
 &=  \exp \left\{ \left[ (\mathcal{L}f_X)(\alpha) - 1\right] \int_s^t \lambda_u \mathrm{d}u \right\},\\ \label{poisson_last}
\end{align}
where $G_{N_t - N_s}$ is the probability generating function of the Poisson random variable with mean $\int_s^t \lambda_u \mathrm{d}u$.
From Equation \eqref{choicekappa} we have thus $\kappa = \frac{1}{\alpha}\left(1 - (\mathcal{L}f_X)(\alpha) \right)$.

In the last step of the proof we consider the discounted share price process, $\left\{\exp\left( -\int_{0}^t r_u \mathrm{d}u \right) S_t, t \geq 0 \right\}$. Note that by design, its expectation for every $t>0$ is finite. Using Corollary \ref{prop_smaller}, we have, for $s < t$,

\begin{align}
\mathbb{E}^\mathbb{Q} \left[S_t^\mathcal{C} S_t^\mathcal{F} \exp \left( -\int_{0}^t r_u \mathrm{d}u  \right)  | \mathcal{G}_s  \right] &= \mathbb{E}^{\mathbb{Q}} \left[ \mathbb{E}^{\mathbb{Q}} \left[ S_t^\mathcal{C} S_t^\mathcal{F} \exp \left( -\int_{0}^t r_u \mathrm{d}u  \right)   | \mathcal{G}_s \vee {\mathcal{F}}_t\right] | \mathcal{G}_s \right] \nonumber\\
&= \mathbb{E}^{\mathbb{Q}} \left[ S_t^\mathcal{F} \exp \left( -\int_{0}^t r_u \mathrm{d}u  \right) \mathbb{E}^{\mathbb{Q}} \left[ S_t^\mathcal{C} | \mathcal{G}_s\right] | \mathcal{G}_s  \right] \nonumber\\
&= S_s^\mathcal{C} \mathbb{E}^{\mathbb{Q}} \left[ S_t^\mathcal{F} \exp \left( -\int_{0}^t r_u \mathrm{d}u  \right) | \mathcal{G}_s \right] \nonumber\\
&= S_s^\mathcal{C} S_s^\mathcal{F}  \exp \left( -\int_{0}^s r_u \mathrm{d}u  \right), \nonumber
\end{align}
where the last line follows from the first step of the proof.
\end{proof}

\noindent We now evaluate the three terms in Equation (\ref{eq:BIG}) separately.

\subsection{Coupon payments}
\label{coup_pmts}
\noindent We consider $\mathbb{E}^\mathbb{Q}[I_1]$, which will be split into the first LIBOR-referencing coupon, and the remaining subsequent ones. One avenue to use in simplifying this term would be to make an approximation to the forward LIBOR rates, and suppose that they are always a constant spread above the riskless rate (see \citet{jarrow2010simple}, for instance). During the periods of financial stability surrounding the 2007/2008 financial crisis, this assumption held to a considerable extent in practice (taking the OIS rate to be the riskless rate). However, in times of crises, empirical work showed that this spread scales up significantly and the assumption becomes questionable \citep{hull2013libor}. We do not use this approximate approach in our analyses but point out this drawback in light of the fact that it has been used in theoretical CAT bond pricing (with the view of obtaining closed-form solutions).

\begin{enumerate}
\item[(i)] The first LIBOR-referencing coupon (which is known at the outset) can be found by risk-neutral valuation formula:
\begin{align}
\mathbb{E}^{\mathbb{Q}} \bigg[ \big( R_{0} + c \big)\Delta Z \mathbb{I}_{\{\tau > t_1\}}B(0, t_1) \bigg], \nonumber
\end{align}
which is, under Assumption 2, equal to
\begin{align}
 \big( R_{0} + c \big) P(r_0,t_1, \theta_r, \sigma_r) \mathbb{P}\left( L_{t_1} < D \right)\Delta Z, \label{eq_term1}
\end{align}
where $R_t$ and $P(0,t)$ are defined in Equations \eqref{R} and \eqref{P}, respectively, $r_0$ is the initial instantaneous interest rate and
\begin{align}
\mathbb{P}\left( L_{t_i} < D \right) &= \exp \left( -\int_0^{t_i} \lambda_u \mathrm{d}u \right) \sum_{n=0}^\infty
 \frac{\left(\int_0^{t_i} \lambda_u \mathrm{d}u \right)^n }{n!}F_X^{n\ast} (D). \label{eqn_cpprob}
\end{align}
where $F_X^{n\ast} (D)$ denotes the $n-$fold convolution of single loss distribution function $F_X$ with itself, evaluated at the positive argument $D$.
\item[(ii)] For the second to $n^{\text{th}}$ LIBOR-referencing coupon, we change measure to the respective forward measure. For each $i \in \{2, ... n \}$, we use the $t_i$ forward measure $\mathbb{Q}^{t_i}$, defined to be the forward measure for the num\'eraire process $P(0, t_i)$; see \citep{bjork2009arbitrage}. Thus the value of the $i^{\text{th}}$ coupon payment at the issue-date is,
\begin{align}
\mathbb{E}^{\mathbb{Q}} \left[ \big(R_{t_{i-1}} + c    \big) \Delta Z \mathbb{I}_{\{\tau > t_i\}}B(0,t_i) \right] &= Zc \mathbb{P} (L_{t_i} < D)\Delta + Z \mathbb{P} (L_{t_i} < D) \mathbb{E}^{\mathbb{Q}} \left[ R_{t_{i-1}}B(0,t_i) \right]\Delta, \label{eq:23}
\end{align}
with the respective probabilities given by Equation (\ref{eqn_cpprob}). Consider $\mathbb{E}^{\mathbb{Q}} \left[ R_{t_{i-1}}B(0,t_i) \right]\Delta$. By changing to the $t_i$ forward measure (see \citet{bjork2009arbitrage}), recalling that the forward LIBOR is a $\mathbb{Q}^{t_i}$ martingale (see \citet{bjork2009arbitrage}) and noting the definition of forward LIBOR (see \citet{brigo2007interest}):
\begin{align}
\mathbb{E}^{\mathbb{Q}} \left[ R_{t_{i-1}}B(0,t_i) \right]\Delta &= P(0, t_i) \mathbb{E}^{\mathbb{Q}^{t_i}} \left[  R_{t_{i-1}}  \right] \Delta \nonumber \\
&= P(0, t_i)R(0, t_{i-1}, t_{i-1} + \Delta) \Delta \\
&= P(0, t_i) \frac{1}{\Delta}\left( \frac{P(0, t_{i-1})}{P(0, t_{i})} -1   \right) \Delta \nonumber \\
&= P(0, t_{i-1}) - P(0, t_i). \label{eq:27}
\end{align}
Inserting Equation (\ref{eq:27}) into Equation (\ref{eq:23}), and adopting our notation for the zero-coupon bond price under the Longstaff models gives the required value at the issue date of the remaining coupon payments, i.e.:

\begin{align}
\mathbb{E}^{\mathbb{Q}} \left[ \sum_{i=2}^N \big(R_{t_{i-1}} + c    \big) \Delta Z \mathbb{I}_{\{\tau > t_i\}}B(0,t_i) \right] &= Z\sum_{i=2}^N \mathbb{P} (L_{t_i} < D) \big[ c \Delta + P(r_0,t_{i-1}, \theta_r, \sigma_r) \nonumber \\  &\quad - P(r_0,t_{i}, \theta_r, \sigma_r) \big].  \label{eq_term2}
\end{align}
\end{enumerate}

\subsection{Redemption amount}
\noindent We now consider $\mathbb{E}^\mathbb{Q}[I_3]$. By analogous reasoning to Section \ref{coup_pmts}, we obtain
\begin{align}
\mathbb{E}^\mathbb{Q} \bigg[Z\mathbb{I}_{\{\tau > T\}} B(0,T) \bigg] &= Z P(r_0,T, \theta_r, \sigma_r) \mathbb{P} (L_T \geq D). \label{eq_term3}
\end{align}

\subsection{Conversion feature}
\label{sec:conversion_feature}
\noindent Finally, we consider $\mathbb{E}^\mathbb{Q}[I_3]$ for a particular type of conversion price. Recall that $K_P$ can either be set as a constant, as the share price at time-of-conversion or be specified as the function $K_P \coloneqq S_\tau^\nu$, for $\nu \in (0,1]$ as specified in Section \ref{section_selection}. The restriction of $\nu$ to this interval is necessary for the Laplace transform of $f_X$ to not be infinite. Notice that if $\nu = 1$, then the conversion price is set equal to the share price.
Our analysis below is general and we simply consider either case of
\begin{enumerate}
\item[4.3.1] $K_P \coloneqq S_\tau^\nu$, or
\item[4.3.2] $K_P \coloneqq K$.
\end{enumerate}
We analyse each in turn below. \\

\noindent {\bf Case 1.} In this case, we are interested in simplifying
\begin{align}
\mathbb{E}^\mathbb{Q} [{S_{\tau}^{1-\nu}}\mathbb{I}_{\{\tau \leq T\}}B(0, \tau) ] \nonumber
\end{align}
\noindent which can be expanded as
\begin{align}
& \mathbb{E}^\mathbb{Q} \bigg[\exp \left( -\alpha (1-\nu) L_\tau + \int_0^\tau (1-\nu)\left( \alpha \kappa \lambda_u - \frac{\sigma_S^2}{2} \right)\mathrm{d}u +  \sigma_S (1-\nu) \tilde{W}_\tau^2 - \nu \int_0^\tau r_u \mathrm{d}u \right) \nonumber \\ &\times S_0 \mathbb{I}_{\{\tau \leq T\}} \bigg]. \label{eqn:next2}
\end{align}

\noindent To evaluate Equation (\ref{eqn:next2}). we will use a specific change of measure. Since $N_t$ in the process $L_t$ is a time-inhomogeneous Poisson process with the intensity $\lambda_t$, then the process
$X_t^\# \coloneqq (t,L_t)$ is a particular case of the so-called piecewise deterministic process (see \citet[Section 5.2]{PalmowskiRolski} for further details). Indeed, in this case one has to consider the empty active boundary, $\Gamma$, take the external state index of this process to be equal to one, and take the state space, $\mathcal{O}_1$, to be $\mathbb{R}^2$ - that
is, $x^\# \coloneqq (t,y)\in \mathbb{R}^2$ is a value of the process $X_t^\#$. Moreover, the differential operator of $X_t^\#$ is specified by $\chi f(t,x^\#) \coloneqq \frac{\partial}{\partial t}f(t,x^\#)$, the transition kernel of the process satisfies $Q(x^\#,\mathrm{d}y)=F_X^\mathbb{Q}(\mathrm{d}y)$ (where $F_X^\mathbb{Q}$ is the severity distribution of the process $L_t$ under $\mathbb{Q}$) and its jump intensity, $\lambda(t,y)$, equals $\lambda_t$.
Now taking $h(y)=e^{-\alpha(1-\nu)y}$ in \citet[Equation (1.1)]{PalmowskiRolski} as a good function and applying it to the generator,
\[\mathcal{A}f(t,x)= \frac{\partial}{\partial t}f(t,x)+\lambda_t\left(\int_0^\infty f(t,x+y)\;F_X(\mathrm{d}y)-f(t,x)\right), \]
of the process $L_t$ (see p. 776 of \citet{PalmowskiRolski})
produces the following exponential martingale:
\begin{align}
\bar{\eta}^{(\nu)} (t)\coloneqq \exp \left( -\alpha (1-\nu) L_t + \varphi (\alpha (1-\nu), t)\right), \label{measure_changeNB}
\end{align}
where
\begin{align}
\varphi (\alpha (1-\nu), t):=[1 - (\mathcal{L}f_X)(\alpha (1 - \nu))]\int_0^t \lambda_u\;\mathrm{d}u. \label{eqn_relation}
\end{align}
and $(\mathcal{L}f_X)(\alpha(1-\nu))$ is the Laplace transform of $f_X$, for the non-negative support of the random variable $X$, evaluated at the argument $\alpha(1-\nu)$.
Following \citet[Theorem 5.3]{PalmowskiRolski} we can now define a new probability measure $\mathbb{P}^{(\nu)}$ via the Radon-Nikodym derivative:
\begin{equation}\label{good}\left.\frac{\mathrm{d}\mathbb{P}^{(\nu)}}{\mathrm{d}\mathbb{P}_C}\right\vert_{\hat{\mathcal{C}}_t}=\bar{\eta}^{(\nu)}(t),\end{equation}
on which our process $L_t$ remains a compound Poisson process but with altered severity distribution and intensity:
\begin{align}
\lambda^{(\nu)}_t &= (\mathcal{L}f_X)(\alpha(1-\nu))\lambda_t , \quad \text{and} \label{eq_intensitychangeNB}\\
F_X^{(\nu)}(\mathrm{d}x) &= \frac{e^{-\alpha(1-\nu) x}F_X^{\mathbb{Q}}(\mathrm{d}x)}{(\mathcal{L}f_X)(\alpha(1-\nu))}. \label{eq_severitychangeNB}
\end{align}
Note that when $\lambda_t=\lambda$ (a constant), then many things simplify. For example,
\begin{align}
\varphi (\alpha (1-\nu), t)&= \lambda \varphi (\alpha (1-\nu))t, \label{eqn_relation2}
\end{align}
\noindent where
\begin{align}
\varphi (\alpha (1-\nu)) &\coloneqq \lambda [1 - (\mathcal{L}f_X)(\alpha (1 - \nu))]. \label{eqn_relation3}
\end{align}

\noindent By applying the measure change specified by Equation \eqref{good} and applying Assumption \ref{ass_2}, we can rewrite Equation (\ref{eqn:next2}) as
\begin{align}
&\mathbb{E}^{\mathbb{Q}_F\otimes \mathbb{P}^{(\nu)}} \bigg[\exp \left( \int_0^\tau (1-\nu)\left( \alpha \kappa \lambda_u - \frac{\sigma_S^2}{2} \right)\mathrm{d}u - \varphi(\alpha(1-\nu), \tau)  +  \sigma_S (1-\nu) \tilde{W}_\tau^2 - \nu \int_0^\tau r_u \mathrm{d}u \right) \nonumber \\& \times S_0 \mathbb{I}_{\{\tau \leq T\}}\bigg] \nonumber \\
&= S_0\int_0^T \; \mathbb{P}^{(\nu)}_\tau (\tau \in \mathrm{d}s) \bigg\{\exp \left(  \int_0^s (1-\nu)\left( \alpha \kappa \lambda_u - \frac{\sigma_S^2}{2} \right)\mathrm{d}u - \varphi(\alpha(1-\nu),s)  \right) \nonumber\\
& \quad \times \mathbb{E}^{\mathbb{Q}_F} \left[ \exp \left( \sigma_S (1 - \nu) \tilde{W}_s^2 - \nu \int_0^s r_u \mathrm{d}u \right)  \right]\bigg\},  \label{eq:end}
\end{align}
\noindent where the last line follows by the independence of  $\tilde{W}_t^1$ and $\tilde{W}_t^2$ from the measure change, and $\mathbb{P}^{(\nu)}_\tau$ is the density function of $\tau$ under $\mathbb{P}^{(\nu)}$. Notice that Equation (\ref{eq:end}) is a product of three terms, the first the distribution of $\tau$ under $\mathbb{P}^{(\nu)}$, the second an exponential term and the third an expectation, under $\mathbb{Q}_F$, of two correlated $\mathbb{Q}$-Brownian motions. We now present the following helpful lemma. Before reading further, recall that in our model we had two correlated standard Brownian motions under $\mathbb{Q}$, $\tilde{W}_t^1$ and $\tilde{W}_t^2$, with correlation coefficient $\rho$. The purpose of this lemma is to remove $\tilde{W}_s^2$ from the expectation in Equation (\ref{eq:end}).  \\

\begin{lemma}
Consider a new probability measure, $\bar{\mathbb{Q}_F}$, given by
the Radon-Nikodym derivative specified by
\begin{align}
\left.\frac{\mathrm{d}\bar{\mathbb{Q}_F}}{\mathrm{d}\mathbb{Q}_F}\right\vert_{\hat{\mathcal{F}}_t}:=\eta^{\ast} (t):= \exp \left((1-\nu)\sigma_S \tilde{W}_t^2 - \frac{1}{2} (1-\nu)^2 \sigma_S^2t  \right) \nonumber
\end{align}
\noindent Then, under the probability measure $\bar{\mathbb{Q}}$ the process $\tilde{\tilde{W}}_t^1 \coloneqq \tilde{W}_t^1 -  \rho \sigma_S (1-\nu)t$ is a standard Brownian motion.
\end{lemma}

\begin{proof}
\noindent Note that by a change of measure from $\bar{\mathbb{Q}}$ to $\mathbb{Q}$, for $\zeta>0$,
\begin{align}
\mathbb{E}^{\bar{\mathbb{Q}}_F} \left[ e^{\zeta \tilde{W}_t^1} \right] &= \mathbb{E}^{\mathbb{Q}_F} \left[ \exp\left(\zeta \tilde{W}_t^1 + (1-\nu)\sigma_S \tilde{W}_t^2 - \frac{1}{2} (1-\nu)^2 \sigma_S^2t \right) \right] \nonumber\\
&= \exp \left( -\frac{1}{2}(1-\nu)^2 \sigma_S^2t  \right) \mathbb{E}^{\mathbb{Q}_F}  \left[ (\zeta,(1-\nu)\sigma_S)\cdot (\tilde{W}_t^1,\tilde{W}_t^2) \right]. \label{eq:end2}
\end{align}
\noindent 
Since the distribution of $(\tilde{W}_t^1,\tilde{W}_t^2)$ is bivariate normal, we can show that Equation (\ref{eq:end2}) is equal to
\begin{align*}
& \exp \left( -\frac{1}{2}(1-\nu)^2 \sigma_S^2t  + \frac{1}{2} \left\{\zeta^2t + 2 \rho \sigma_S (1-\nu)\zeta t + (1-\nu)^2\sigma_S^2t  \right\} \right) \\
&= \exp\left( \frac{1}{2}\zeta^2t + \rho \sigma_S (1-\nu)\zeta t \right).
\end{align*}

\end{proof}

\noindent In consequence, we are able to proceed with simplifying
\begin{align}
\mathbb{E}^{\mathbb{Q}_F} \left[ \exp \left( \sigma_S (1 - \nu) \tilde{W}_s^2 - \nu \int_0^s r_u \mathrm{d}u \right)  \right]  \label{eq:bmexp}
\end{align}
from Equation (\ref{eq:end}). Now, change measure from $\mathbb{Q}_F$ to $\bar{\mathbb{Q}}_F$ to obtain that Equation (\ref{eq:bmexp}) is equal to
\begin{align}
\exp\left( \frac{1}{2} (1-\nu)^2 \sigma_S^2 s \right) \mathbb{E}^{\bar{\mathbb{Q}}_F} \left[ \exp \left( - \nu \int_0^s r_u \mathrm{d}u \right)  \right] \label{eq:finaleq}
\end{align}
\noindent where $\{r_t, t \geq 0\}$ is the interest-rate process now under the measure $\bar{\mathbb{Q}}_F$. Theorem \ref{prop_DSR} reminds us that the Longstaff interest-rate model specified in (\ref{eq:9}), under $\mathbb{Q}_F$, remains a Longstaff interest-rate model under $\bar{\mathbb{Q}}_F$ with the two parameters respectively given by
\begin{align*}
{\theta}_r^\ast = \theta_r - \sigma_r \rho \sigma_S (1 - \nu) \; ; \; {m}_r^\ast = \frac{m_r \theta_r}{\theta_r - \sigma_r \rho \sigma_S (1-\nu)}.
\end{align*}

\noindent It can also be shown under the measure $\bar{\mathbb{Q}}$, by a simple application of It\^o's Lemma, that the process $\{\tilde{r}_t, t \geq 0 \}$, where $\tilde{r}(t) \coloneqq \nu r(t) $ is a Longstaff interest-rate process with $\bar{\mathbb{Q}}$-dynamics given by
\begin{align}
\mathrm{d} \tilde{r}_t &= {\theta}_r^\circ \left({m}_r^\circ - \sqrt{\tilde{r}_t}\right)\mathrm{d}t + {\sigma}^\circ_r \sqrt{\tilde{r}_t}\mathrm{d}\tilde{\tilde{W}}^1_t, \nonumber
\end{align}
\noindent where
\begin{align}
\theta_r^\circ &= \sqrt{\nu}{\theta}_r^\ast \nonumber\; ; \;
m_r^\circ = \sqrt{\nu}{m}_r^\ast \; ; \;
\sigma_r^\circ = \sqrt{\nu}\sigma_r, \nonumber
\end{align}
\noindent and furthermore that the property that $m_r^\circ = \nicefrac{\left({\sigma}_r^\circ \right)^2}{4\theta_r^\circ}$ remains. Thus, we can rewrite Equation (\ref{eq:finaleq}) as
\begin{align}
\exp\left( \frac{1}{2} (1-\nu)^2 \sigma_S^2 s \right) \mathbb{E}^{\bar{\mathbb{Q}}} \left[ \exp \left( - \int_0^s \tilde{r}_u \mathrm{d}u \right)  \right] \label{eq:finaleq2}
\end{align}
and we have, in consequence, a solution to the partial differential equation, specified by (\ref{eqn_pde}), multiplied by a constant. Using Equation (\ref{ZCBprice}), we can simplify Equation (\ref{eq:finaleq2}) and obtain that it equals to
\begin{align}
\exp\left( \frac{1}{2} (1-\nu)^2 \sigma_S^2 s \right) P(r_0, s, \theta_r^\circ, \sigma_r^\circ). \label{eq:finaleq3}
\end{align}
Hence, finally, Equation (\ref{eq:end}) can be expressed as an integral involving an infinite sum, that is
\begin{align}
& S_0\int_0^T \; \mathbb{P}^{(\nu)}_\tau (\tau \in \mathrm{d}s) \exp \left(\int_0^s (1-\nu)\left( \alpha \kappa \lambda_u - \frac{\sigma_S^2}{2} \right)\mathrm{d}u - \varphi(\alpha(1-\nu), s) + \frac{1}{2} (1-\nu)^2 \sigma_S^2s \right) \nonumber\\
& \quad \times P(r_0, s, \lambda, \theta_r^\circ, \sigma_r^\circ) \label{eqn:highlyintuitive2}\\
&= S_0 \int_0^T A(s) \; P(r_0, s, \theta_r^\circ, \sigma_r^\circ) \; B(s)  \; \mathrm{d}s, \label{eqn_finalprice2}
\end{align}
\noindent where
\begin{align*}
A(s) &= \exp{\left( -(\sigma_S^2/2)\nu (1 - \nu)s + (1- \nu) [\varphi (\alpha, s) -  \varphi (\alpha(1-\nu),s)] - \int_0^s\lambda^{(\nu)}_u \mathrm{d}u \right)} \\
B(s) &= \lambda^{(\nu)}_s \sum_{n=0}^\infty\left[ \frac{\left(\int_0^s\lambda^{(\nu)}_u \mathrm{d}u\right)^{n-1}}{(n-1)!} - \frac{\left(\int_0^s\lambda^{(\nu)}_u \mathrm{d}u\right)^{n}}{n!} \right] F_X^{(\nu)n\ast} (D)
\end{align*}
and where $F_X^{(\nu)n\ast} (D)$ denotes the $n-$fold convolution of $F_X^{(\nu)}$ with itself, evaluated at the argument $D$, and $\varphi( \alpha, s) \coloneqq \varphi( \alpha(1- \nu), s)|_{\nu =0}$. Recall that $F_X^{(\nu)}$ and $\lambda^{(\nu)}_t$ are given in \eqref{eq_severitychangeNB} and \eqref{eq_intensitychangeNB}, respectively
and $P$ in \eqref{ZCBprice}.

We are now in a position to state the time-zero risk-neutral price of an IL CocoCat in analytical form. We formalise this all in Theorem \ref{theorem_3}.

\begin{theorem} \label{theorem_3}
In an analytical form, the time-zero risk-neutral price, $V_0$, of an IL CocoCat with conversion price equal to $S_\tau^\nu$ for $\nu \in (0, 1]$, and assuming the dynamics given in Equations (\ref{eq:9}) to (\ref{eq:13}), is given by
\begin{align}
V_0 &= Z \left(I_1^E + I_2^E + I_3^E\right),
\end{align}
where
\begin{align*}
I_1^E &=  \big( R_{0} + c \big) P(r_0,t_1, \theta_r, \sigma_r) \mathbb{P}\left( L_{t_1} < D \right)\Delta  + \sum_{i=2}^N \mathbb{P} (L_{t_i} < D) \left[ c \Delta + P(r_0,t_{i-1}, \theta_r, \sigma_r) - P(r_0,t_{i}, \theta_r, \sigma_r)  \right]  ;\\
I_2^E &= S_0 \zeta \int_0^T A(s) \; P(r_0, s, \tilde{\theta}_r, \tilde{\sigma}_r) \; B(s)  \; \mathrm{d}s    ;\\
I_3^E &= P(r_0, T, \tilde{\theta}_r, \tilde{\sigma}_r)  \mathbb{P} (L_T \geq D)
\end{align*}
and $I_1^E$ is identified in Equations \eqref{eq_term1} and \eqref{eq_term2},
$I_2^E$ in Equation \eqref{eqn_finalprice2} and
$I_3^E$ in Equation \eqref{eq_term3}.
\end{theorem}

\noindent {\bf Case 2.} In the final part of this section, we consider the interesting case pertaining to when $K_P$ is a constant\footnote{It is not possible to use Theorem \ref{theorem_3} to deduce the result, since it is impossible to analytically evaluate $P$ in Equation (\ref{eqn_finalprice2}).} of $K$. Therefore, to analyze $\mathbb{E}^\mathbb{Q}[I_3]$ we will only concern ourselves with simplifying
\begin{align}
\mathbb{E}^\mathbb{Q} [S_{\tau }\mathbb{I}_{\{\tau \leq T\}}B(0,\tau) ] \nonumber
\end{align}
\noindent which can be expanded as
\begin{align}
&\mathbb{E}^\mathbb{Q} \bigg[S_0 \exp \left(  -\alpha L_\tau + \alpha \kappa \int_0^\tau \lambda_u \mathrm{d}u + \int_0^\tau r_u \mathrm{d}u - \frac{\sigma_S^2}{2}\tau + \sigma_S \tilde{W}_\tau^2 \right)  \mathbb{I}_{\{\tau \leq T\}} \exp \left(- \int_0^\tau r_r \mathrm{d}u \right) \bigg] \nonumber \\
&= \mathbb{E}^\mathbb{Q} \bigg[S_0 \mathbb{I}_{\{\tau \leq T\}} \exp \left( -\alpha L_\tau + \int_0^\tau \left( \alpha \kappa \lambda_u - \frac{\sigma_S^2}{2} \right)\mathrm{d}u +  \sigma_S \tilde{W}_\tau^2\right) \bigg]. \label{eqn:four}
\end{align}

\noindent Since $L_t$ is a time-inhomogeneous compound Poisson process, then the measure change considered in Section \ref{sec:conversion_feature} part (i) can be applied with $\nu=0$. We denote $\eta (t)=\eta^{(0)}(t)$ for the density process \eqref{measure_changeNB}.
Note that $\varphi (\alpha, t) = \varphi(\alpha (1-\nu), t)|_{\nu = 0}$. \\

\noindent Upon consideration of this measure change, notice that Equation (\ref{eqn:four}) can be rewritten as
\begin{align}
&\mathbb{E}^\mathbb{Q} \bigg[S_0 \mathbb{I}_{\{\tau \leq T\}} \exp \left( -\alpha L_\tau + \varphi (\alpha, \tau) - \varphi (\alpha, \tau) + \int_0^\tau \left( \alpha \kappa \lambda_u - \frac{\sigma_S^2}{2} \right)\mathrm{d}u +  \sigma_S \tilde{W}_\tau^2\right) \bigg] \nonumber \\
&= \mathbb{E}^{\mathbb{Q}_F\otimes\mathbb{P}^{(0)}} \bigg[S_0 \mathbb{I}_{\{\tau \leq T\}} \exp \left( \int_0^\tau \left( \alpha \kappa \lambda_u - \frac{\sigma_S^2}{2} \right)\mathrm{d}u -  \varphi (\alpha, \tau)  +  \sigma_S \tilde{W}_\tau^2\right) \bigg] \nonumber \\
&= S_0\int_0^T \; \mathbb{P}^{(0)}_\tau (\tau \in \mathrm{d}s) \exp \left( \int_0^s \left( \alpha \kappa \lambda_u - \frac{\sigma_S^2}{2} \right)\mathrm{d}u -  \varphi (\alpha, s) \right)  \mathbb{E}^{\mathbb{Q}_F} [ e^{\sigma_S \tilde{W}_s^2}  ], \label{eq_next}
\end{align}
\noindent where the last line follows because $\tilde{W}_s^2$ is independent of the measure change for all $s \in [0,T]$, and $\mathbb{P}^{(0)}_\tau$ is the density function of $\tau$ under $\mathbb{P}^{(0)}$.  By considering the moment generating function of $\tilde{W}_t^2$, we obtain that Equation (\ref{eq_next}) is equal to
\begin{align}
&S_0\int_0^T \; \mathbb{P}^{(0)}_\tau (\tau \in \mathrm{d}s) \exp \left( \alpha \kappa \int_0^s \lambda_u\mathrm{d}u -  \varphi (\alpha,s) \right) \nonumber \\
&= S_0\int_0^T \; \mathbb{P}^{(0)}_\tau (\tau \in \mathrm{d}s) \nonumber \\
&= S_0 \mathbb{P}^{(0)}_\tau (\tau \leq T) \label{highly_intuitive}\\
&= S_0 \left[1 - \exp\left(-\int_0^T \lambda_u^{(0)}\mathrm{d}u\right)\sum_{n=0}^{\infty}\frac{\left(\int_0^T \lambda_u^{(0)}\mathrm{d}u \right)^n}{n!} F_X^{(0)n\ast} (D)  \right], \label{eqn_finalprice}
\end{align}
\noindent where $F_X^{(0)n\ast} (D)$ denotes the $n-$fold convolution of $F_X^{(0)}$ with itself, evaluated at the argument $D$. Notice that the simplification of the expression, $\mathbb{E}^\mathbb{Q} [S_{\tau }\mathbb{I}_{\{\tau \leq T\}}B(0,\tau) ]$, as given by Equation (\ref{highly_intuitive}) is highly intuitive. The discounted value of the share price at the time of trigger is simply the existing share price, $S_0$, multiplied by the probability of trigger under a probability measure that explicitly adjusts for the aggregate loss process.

We can now give the time-zero risk-neutral price of an IL CocoCat with constant conversion price in analytical form. We present this in Theorem \ref{theorem_4}.

\begin{theorem} \label{theorem_4}
In an analytical form, the time-zero risk-neutral price, $V_0$, of an IL CocoCat with constant conversion price equal to $K$, and assuming the dynamics given in Equations (\ref{eq:9}) to (\ref{eq:13}), is given by
\begin{align}
V_0 &= Z \left(I_1^E + I_2^E+ I_3^E\right).
\end{align}
where $I_1^E$ is identified in Equations \eqref{eq_term1} and \eqref{eq_term2},
$I_2^E$ in Equation \eqref{eqn_finalprice} (multiplied by $\nicefrac{\zeta}{k}$) and
$I_3^E$ in Equation \eqref{eq_term3}.
\end{theorem}

\subsection{Remarks on Theorems \ref{theorem_3} and \ref{theorem_4}}
\noindent We now provide a brief motivation behind the use of Equations Theorems \ref{theorem_3} and \ref{theorem_4} for calculating IL CocoCat prices. Compared to the use of Monte Carlo simulation for approximating Equation (\ref{eq:BIG}) directly, the use of our Theorems is more accurate. Firstly, if the convolutions of the losses under the measure $\mathbb{P}^{(\nu)}$ and $\mathbb{P}^{(0)}$ are known in closed form, we can approximate the price of the IL CocoCat (and in particular the value of the conversion feature) directly without the need for Monte Carlo simulation. Hence, without any simulation we will be able to compute Equation (\ref{eqn_finalprice2}) by truncating the summation at a computationally-suitable upper limit, and by discretising the integral. Moreover, without any simulation we will be able to compute Equation (\ref{eqn_finalprice}) by truncating the summation at a computationally-suitable upper limit). Such a case can indeed arise if we assume that, under the probability measure $\mathbb{P}$, the losses follow a gamma distribution\footnote{The use of Equations (\ref{eqn:highlyintuitive2}) and (\ref{highly_intuitive}) requires first an exponential tilting of the loss random variable, and thereafter an $n$-fold convolution. Now, an exponentially-tilted gamma distributed random variable is again gamma distributed, and so too is its convolution. Moreover, the Laplace transform of a gamma random variable exists in closed-form.}.

Secondly, for more heavy-tailed distributional assumptions of the loss random variables, it is necessary to use Monte Carlo simulation to numerically evaluate the integrals in Equations (\ref{eqn:highlyintuitive2}) and (\ref{highly_intuitive}). That is, one has to simulate the empirical distribution of the stopping times $\tau$'s. This is still more accurate, and in fact faster, than evaluating Equation (\ref{eq:BIG}) directly by Monte Carlo since in the former case only one process, $L_t$, has to be simulated compared to the latter, where it is necessary to simulate three processes, namely $L_t$, $r_t$ and $S_t$.

\section{Empirical illustration}
\label{section_empirics}

\noindent In this section we present numerical experiments as a first foray into the price behaviour of the IL CocoCat. Gaining an understanding of the IL CocoCat price behaviour, for varying parameters, is crucial in the design stage of such an instrument. Moreover, it could be instrumental in preparing illustrations for pitch books which ILS structurers could use to market new ILS issue types to potential issuers.

In our simulations we pay particular attention to the conversion feature as this will be the item of most interest when issuing the instrument. Also, since exact closed-form solutions are not available for the underlying loss severity distributions we choose, we use Monte Carlo simulation based on importance sampling, for the loss process. $100,000$ simulations are used in each respective instance after applying the importance sampling.

In order to obtain numerical values for the IL CocoCat prices, we need to specify a base set of parameters. Such parameters are specified in Table \ref{Table_params}.

\begin{table}[H]
  \center
      \caption{Selected parameter values for the IL CocoCat price numerical illustration.}
  \label{Table_params}
    \begin{tabular}{l|l|c}
    \toprule
          \multicolumn{2}{l}{\textbf{\small{Compound Poisson loss process parameters}}} & \textbf{\small{Values}} \\
    \midrule
		\small{$\lambda_t \qquad$} & \small{Intensity of catastrophe loss process} & \small{Equation (\ref{eq:6})}\\
		\small{$c_b$} & \small{Shape parameter of Burr severity distribution} & \small{$1.57$}\\
		\small{$k_b$} & \small{Shape parameter of Burr severity distribution} & \small{$0.7$}\\
		\small{$\zeta_b$} & \small{Scale parameter of Burr severity distribution} & \small{$9.53 \times 10^7$}\\
    \midrule
     \multicolumn{3}{l}{\textbf{\small{Interest-rate process parameters under $\mathbb{Q}$}}}  \\
    \midrule
    \small{$r_0$} & \small{Initial instantaneous interest rate} & \small{$0.02$}\\
    \small{$\theta_r$} & \small{Model parameter} & \small{$0.02$}\\
    \small{$\sigma_r$} & \small{Instantaneous volatility} & \small{$0.1$}\\
    \small{$m_r$} & \small{Model parameter} & \small{$1.125 \times 10^{-3}$}\\
    \midrule
     \multicolumn{3}{l}{\textbf{\small{Issuer's share price process parameters under $\mathbb{P}$}}}  \\
        \midrule
    \small{$S_0$} & \small{Initial share price} & \small{$10$}\\
    \small{$\rho$} & \small{Correlation coefficient of share and interest-rate processes} & \small{$-0.5$}\\
    \small{$\alpha$} & \small{Effect of losses on log share price} & \small{$5.81 \times 10^{-11}$}\\
        \midrule
     \multicolumn{3}{l}{\textbf{\small{IL CocoCat parameters}}}  \\
        \midrule
        \small{$K_P$} & \small{Constant conversion price} & \small{Varies}\\
        \small{$\nu$} & \small{Power parameter for conversion price} & \small{Varies}\\
        \small{$\zeta$} & \small{Conversion fraction} & \small{$0.2$}\\
        \small{$\Delta$} & \small{Tenor: time between coupon payments} & \small{$0.25$}\\
        \small{$c$} & \small{Constant spread (CAT risk premium)} & \small{$0.1$}\\
        \small{$Z$} & \small{Nominal amount} & \small{$1$}\\
        \small{$T$} & \small{Term} & \small{Varies}\\
        \small{$D$} & \small{Threshold level} & \small{Varies}\\
    \bottomrule
    \end{tabular}
\end{table}

Firstly, it is necessary to select parameters for the compound Poisson process underlying the IL CocoCat. We do so by using the following intensity, which was fitted to PCS data from 1985 to 2011 by \citet{giuricich}:
\begin{align}
\lambda(t) = 24.93 + 0.03t + 5.61 \sin\{2\pi (t+7.07)\} + 0.30 \exp \left\{ \cos \left( \frac{2 \pi t}{4.76} \right)\right\}. \label{eq:6}
\end{align}

However, we also point out the following interesting, potential application of the work by \citet{braun}. Suppose that we only modelled an index consisting of losses from earthquakes and/or hurricanes\footnote{With a more detailed breakdown of PCS data, it is indeed possible to extract this information.}. If we did not assume that for the catastrophe-risk variables the real-world and risk-neutral probability measures coincided\footnote{That is, if we did not use Assumption \ref{ass_2} but made the slightly weaker assumption that a random variable retains its distributional characteristics when moving from the real-world to the risk-neutral probability measure.}, and rather used a risk-neutral measure for the catastrophe-risk variables, then we could use Braun's implied intensities backed out from catastrophe swap transactions. In doing so, we could ensure that the CocoCat is consistently priced with other instruments in the catastrophe risk markets.

Moreover, we want our process to be based on a heavy-tailed underlying severity distribution, so that low-frequency and high-severity disasters are indeed accounted for. Evidence for heavy-tailed underlying distributional properties in catastrophe-related economic and insured losses has been found by \citet{levi1991statistical, burnecki2000property, milidonis2008tax, braun, mm} and \citet{giuricich}. In all endeavours to assess which heavy-tailed distribution fits such data, it appears from a review of the literature on catastrophe-related ILS that the Burr distribution consistently comes out best. In this section, we indeed use the Burr type XII distribution with probability density function specified by
\begin{align*}
f(x, \zeta_b, c_b, k_b) = \frac{\frac{k_bc_b}{\zeta_b}\left( \frac{x}{\zeta_b} \right)^{c_b-1}}{\left( 1 + \left( \frac{x}{\zeta_b} \right)^{c_b} \right)^{k_b+1}} \quad \zeta_b, c_b, k_b > 0 \quad \text{and} \quad x > 0,
\end{align*}
\noindent where $c_b$ and $k_b$ are the shape parameters and $\zeta_b$ is the scale parameter. For illustrative purposes, we consider the Burr distribution  fitted  by \citet{giuricich} and take $c_b=1.57$, $k_b=0.7$ and $\zeta_b = 9.53 \times 10^7$. In passing, note that it would be possible to add risk margins to our various parameters of the loss process derived from real-world data. This is in line with some of the traditional actuarial approaches (see for example the recent work by \citet{chang2017integrated} for a brief discussion on this point, and also \citet{braun} for a brief mention), but as always, such a choice is subjective and specific to the modeller. Finally, it must be noted that severity distributions (and also intensity functions) fitted to both real-world and synthetically-simulated data from sophisticated natural catastrophe models can also be employed in our context. In discussion with industry practitioners it would appear that this is the approach used to price natural catastrophe-related instruments in practice.

Secondly, we consider the interest-rate parameters. We reiterate that simulation of the interest-rate process is not necessary given the simplification of the IL CocoCat pricing formulae in Theorems \ref{theorem_3} and \ref{theorem_4}. We let the risk-neutral parameters, $\theta_r$ and $\sigma_r$, equal $0.2$ and $0.03$ respectively. Moreover, we specify the correlation coefficient, $\rho$, to be $-0.5$. All of these parameters are in line with \citet{lo2013valuation} and \citet{wang2016catastrophe}, as well as previous literature concerning the modelling of interest-rates in the context of insurance (see, for example, \citet{duan2002maximum} and \citet{chang2011valuation}).

Thirdly, we consider the share price process' parameters. We set the initial share price, $S_0$, to be $10$. The effect of the catastrophic losses on the logarithm of the share price, $\alpha$, is found in a similar fashion to \citet{jaimungal2006catastrophe}. We let $\alpha$ represent the percentage drop in the log share price per unit of expected loss, that is
\begin{align*}
\alpha &= \frac{\delta}{E^{\mathbb{P}}[X_k]},
\end{align*}
\noindent and we consider the case, as in \citet{jaimungal2006catastrophe}, where $\delta = 0.02$. In consequence, $\alpha = 5.81 \times 10^{-11}$. Despite $\alpha$ being very small in size its effect is still prevalent since it is multiplied by $L_t$ in Equation (\ref{eq:15}), which is relatively much larger than $\alpha$.

Finally, we give thought to the various parameters for the IL CocoCat itself. We set the contractually-specified conversion fraction, $\zeta$, equal to $20$\%, which is in line with previous literature on CocoCats (see \citet{georgiopoulos2016valuation}), and we let the tenor be 3 months (in line with \citet{jarrow2010simple}). For illustrative purposes we let the IL CocoCat spread be $10$\%, and set the nominal value of the bond, $Z$, equal to 1. Also, we will let certain parameters vary, namely the term ($T$), the threshold level ($D$) and $\nu$, in order to assess how the IL CocoCat time-zero price changes with changes in these parameters.

We now comment on how one can estimate numerical values for Equations (\ref{eqn:highlyintuitive2}) and (\ref{highly_intuitive}) via Monte Carlo simulation under the measures $\mathbb{P}^{(\nu)}$ and $\mathbb{P}^{(0)}$ respectively. In order to evaluate these integrals, it is necessary to develop empirical distributions for the stopping time $\tau$ under the respective measures by simulating the paths for the process $L_t$ by considering the value of $L_t$ for each $t<T$. To simulate paths for $L_t$ under the measures $\mathbb{P}^{(\nu)}$ and $\mathbb{P}^{(0)}$, it is necessary to know the intensity and severity distributions under these measures and the necessary links are provided by Equations (\ref{eq_intensitychangeNB}) and (\ref{eq_severitychangeNB}) respectively. Equation (\ref{eq_intensitychangeNB}) is easy to find given that we numerically know the Laplace transform of the severity random variable. However, coping with Equation (\ref{eq_severitychangeNB}) is not immediately obvious for Burr-distributed severity random variables. This is primarily because the exponentially-tilted Burr distribution does not have a density that is known and computable. So, to simulate losses from Equation (\ref{eq_severitychangeNB}), we employed the acceptance-rejection algorithm \cite{von195113}. We point out that such a simulation technique can be used in any setting where the measure change is specified by Equation (\ref{eq_severitychangeNB}). Note that such a simulation technique\footnote{Note that exponential tilting is a case in point.} is common in the sphere of simulating random variables, under other probability measures, from those random variables under a given, complex measure change  \cite{asmussen2007stochastic}. We proceed as follows: if $f(x, \zeta_b, c_b, k_b)$ is the pdf of the Burr-distributed losses under $\mathbb{Q}$ (which is invertible and known), then we can generate from the density $f^{(\nu)}(x, \zeta_b, c_b, k_b)$ by generating first from $f(x, \zeta_b, c_b, k_b)$ directly\footnote{To simulate efficiently from the heavy-tailed density $f(x, \zeta_b, c_b, k_b)$, we use importance sampling (see \citet{giuricich}).}, and then applying the acceptance-rejection algorithm with the constant $c_R \coloneqq [(\mathcal{L}f_X)(\alpha(1-\nu))]^{-1}$.  Under the measure $\mathbb{P^{(\nu)}}$, the simulation algorithm is summarised in the following pseudo-code. The algorithm for $\mathbb{P}^{(0)}$ is analogous. \\

\hrule
\vspace{2mm}
\noindent \textbf{Key steps of algorithm to generate $L_t$}
\vspace{2mm}
\hrule
\begin{enumerate}
\item[\textbf{1.}] For the interval $t_{i-1}$ to $t_i$, generate a Poisson realisation, $n_i$, with intensity $\int_{t_{i-1}}^{t_i} \lambda_u \mathrm{d}u$.
\item[\textbf{2.}] Generate $n_i$ realisations from $f^{(\nu)}(x, \zeta_b, c_b, k_b)$ by the acceptance-rejection algorithm:
\begin{enumerate}
\item[\textbf{(i)}] Generate $U$ from $\mathcal{U}(0,1)$ and independent $X$ with density $f(x, \zeta_b, c_b, k_b)$.
\item[\textbf{(ii)}] If $U < \frac{f^{(\nu)}(X, \zeta_b, c_b, k_b)}{c_Rf(X, \zeta_b, c_b, k_b)}$ then accept $X$ else return to step (i).
\item[\textbf{(ii)}] Continue until $n_i$ realisations from $f^{(\nu)}(x, \zeta_b, c_b, k_b)$ are drawn.
\end{enumerate}
\item[\textbf{3.}] Continue until time $t_N$, consecutively adding the accepted realisations.
\end{enumerate}
\hrule
\vspace{5mm}

Before presenting the numerical results, we now give a brief overview of the remainder of Section \ref{section_empirics}. The aspect of the IL CocoCat which will matter most will be the conversion feature (i.e. $I_2$ from Theorem \ref{theorem_3}), so we will focus on the analysis of different types thereof. Two broadly different conversion prices are considered for the conversion feature and are each evaluated by Monte Carlo simulation. Probabilities relating to the loss process in Equations (\ref{eq_term1}), (\ref{eq_term2}) and (\ref{eq_term3}) (i.e. $\mathbb{P}(L_{t_i} < D) = \mathbb{P}(\tau >= t_i)$)  are also evaluated by simulating the distribution of the stopping times via Monte Carlo. However, the remaining terms relating to the interest-rate process in Equations (\ref{eq_term1}), (\ref{eq_term2}) and (\ref{eq_term3}) as well as in Equation (\ref{eqn:highlyintuitive2}) are evaluated via the closed-form solutions available. So, Section \ref{sec:p1} considers the case when the IL CocoCat has a constant conversion price of $K_P$, while Section \ref{sec:p2} looks at the case when the conversion price is a function of the share price, that is $K_P = S_\tau^\nu$ for $\nu \in [0,1)$. The latter includes the special case when $K_P = S_\tau$, which is of interest since the conversion price is set to the share price at time-of-conversion. Finally, Section \ref{sec:p3} compares the price behaviour of the IL CocoCat across three cases: when $K_P$ is a constant, $K_P = S_\tau$ and $K_P = S_\tau^{0.5}$.

\subsection{Case 1: $K_P$ is a constant}
\label{sec:p1}
\noindent We begin by studying the behaviour of the IL CocoCat price in the context of a changing constant conversion price ($K_P$) and threshold level ($D$). This is illustrated in Figure \ref{fig_KP} panel (a). Note that, by the design of Equations (\ref{eq:9}) to (\ref{eq:11}), the interest-rate process does not affect the value of this IL CocoCat's conversion feature, $(\zeta /K_P)\mathbb{E}^\mathbb{Q} [ S_{\tau}Z\mathbb{I}_{\{\tau \leq T\}}B(0,\tau)]$, at all.

\begin{figure}
\begin{center}
\includegraphics[width=16.5cm, height=7cm]{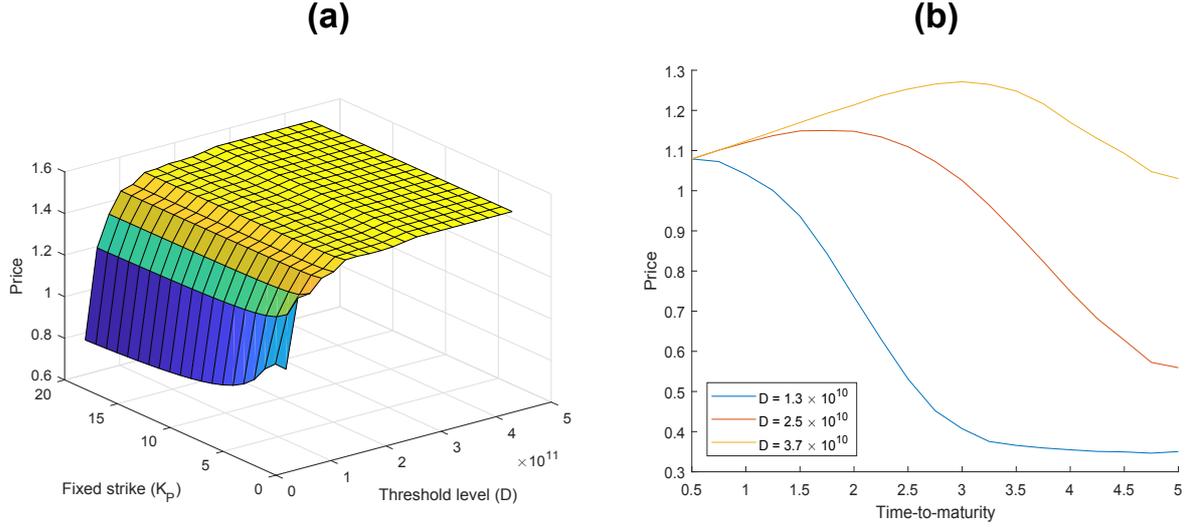}
\caption{(a) Price of IL CocoCat as a function of conversion price ($K_P$) and threshold level ($D$), calculated using $100,000$  Monte Carlo simulations of the loss process, $L_t$. (b) Price of IL CocoCat as a function of term ($T$) for three different threshold levels, calculated using $100,000$  Monte Carlo simulations of the loss process, $L_t$, taking $K_P = 8$.} \label{fig_KP}
\end{center}
\end{figure}

From Figure \ref{fig_KP} panel (a), it is clear that for all constant strike prices, the IL CocoCat prices level out around a value of approximately $1.5$ per unit nominal as the threshold level increases. Intuitively this makes sense as chances of trigger at such high threshold levels become smaller and smaller, making the conversion feature less valuable and the coupon and redemption payments more valuable. In fact, for these high threshold levels the IL CocoCat appears to behave like a corporate bond (ignoring credit risk). However, for lower threshold levels the conversion feature which comprises a partial write down of the principal invested, is more valuable compared to the redemption (or principal) amount and so IL CocoCat prices are lower.

It is also interesting to study the IL CocoCat price behaviour over different terms, the numerical results of which are shown in Figure \ref{fig_KP} panel (b). It shows that the longer the term, the lower the price of the IL CocoCat. This is because for longer terms, there is a greater probability of trigger and in consequence there is a higher risk of an investor losing $1-\zeta$ per unit nominal. But as expected, this decline in IL CocoCat price is slower for those with higher threshold levels.

\subsection{Case 2: $K_P = S_\tau^\nu$}
\label{sec:p2}
\noindent We now study the case when the conversion price is a function of the share price at time-of-conversion, in the context of a changing value of $\nu$ and threshold level $D$. Figure \ref{fig_KP_nu} panel (a) shows this.  But also, note that the changing interest-rate does have an impact on the price of an IL CocoCat structured like this, so we endeavour to study this effect. The results of this investigation are shown in Figure \ref{fig_KP_nu} panel (b), whereby we analyse the effect of altering the two interest-rate process parameters, $\sigma_r$ and $\theta_r$, on the IL CocoCat price for different terms.

\begin{figure}
\begin{center}
\includegraphics[width=16.5cm, height=7cm]{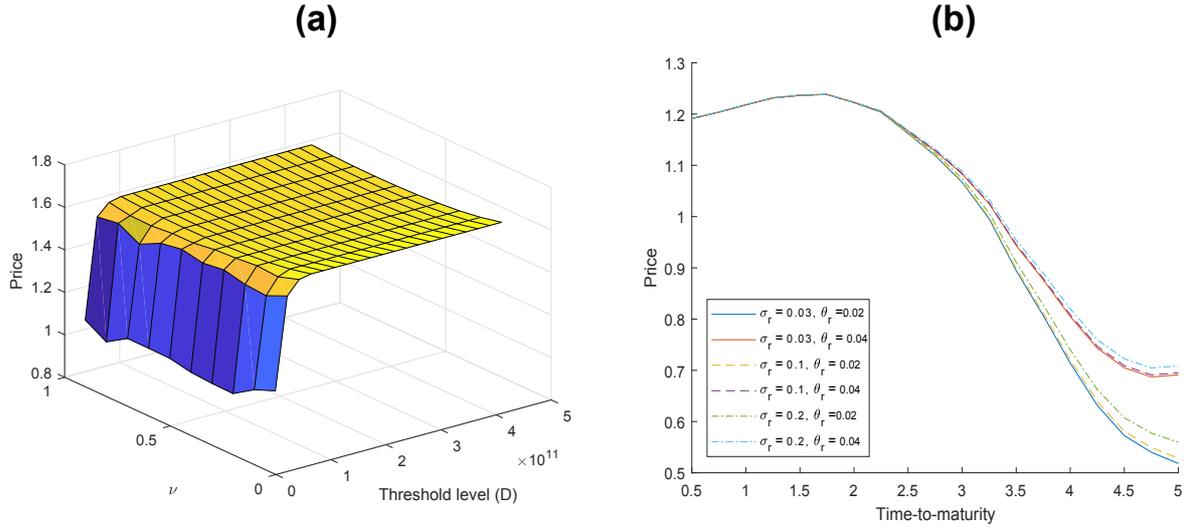}
\caption{(a) Price of IL CocoCat as a function of $\nu$ and threshold level ($D$), calculated using $100,000$  Monte Carlo simulations of the loss process, $L_t$. (b) Price of IL CocoCat as a function of term ($T$) for three different interest-rate volatilities and two different values for the parameter $\theta_r$, calculated using $100,000$  Monte Carlo simulations of the loss process, $L_t$, taking $\nu = 0.5$.} \label{fig_KP_nu}
\end{center}
\end{figure}

Consider Figure \ref{fig_KP_nu} panel (a). As a function of $\nu$ and threshold level, similar behaviour is noted to the case when the conversion feature is set at a constant level (as in Case I). Given that the conversion fraction $\zeta$ is relatively small at $0.2$, the results seen make intuitive sense as the conversion fraction has a small effect on the overall price of the IL CocoCat. The time-zero value of the coupons and principal amount  have a markedly greater impact on the price.

If $\zeta$ is increased in size, of which it can be, then the IL CocoCat price behaviour is observed to change and the time-zero value of the conversion price has a visible effect on increasing the price of the IL CocoCat. In fact, from our numerical simulations illustrated in Figure \ref{fig_KP_zeta}, it is clear that from values for $\zeta$ of $0.8$ and higher, higher time-zero prices for the IL CocoCat can arise. However, we do remark that such a high value for $\zeta$ is potentially unrealistic in practice, since not much of the IL CocoCat will be available for use by the issuer as capital relief. Moreover, on the basis of other numerical analyses which we performed it was interesting to note that the lower the conversion price (i.e. the smaller $\nu$ is), the greater the effect of increasing $\zeta$ on the time-zero IL CocoCat price.

\begin{figure}
\begin{center}
\includegraphics[width=9cm, height=6.5cm]{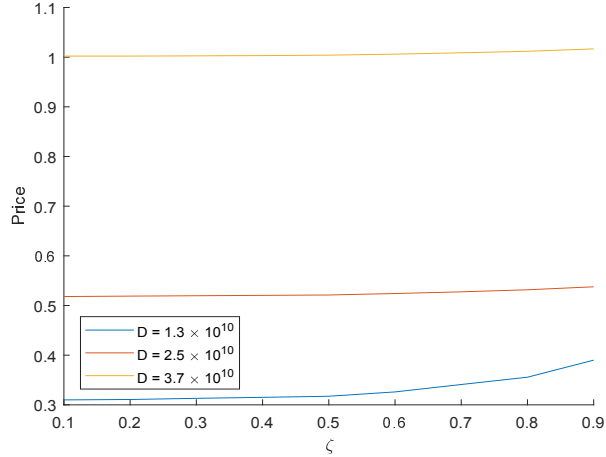}
\caption{(a) Price of IL CocoCat as a function of $\zeta$ and three different threshold levels ($D$), calculated using $100,000$  Monte Carlo simulations of the loss process, $L_t$, and taking $\nu = 0.5$.} \label{fig_KP_zeta}
\end{center}
\end{figure}

Furthermore, we studied the effect of changing interest-rates on the price of the IL CocoCat, by analysing the effects of changing instantaneous volatility and the parameter $\theta_r$ in Figure \ref{fig_KP_nu} panel (b). Only for medium to long terms, did changing interest-rates have a visible effect, with the parameter $\theta_r$ being the greatest influence leading to this change. Note that at longer terms, the different interest-rates had a large impact on the time-zero values of the coupons and principal amount, hence leading to the different values observed for the IL CocoCat time-zero price for these terms.

We close this section with the following remark. From our numerical simulations, it seems evident that the IL CocoCat price is most sensitive to the threshold level, $D$. This is particularly the case for lower threshold levels, but for higher threshold levels (i.e. those in excess of $10^{10}$), it does not make much of a difference. The latter case is true since the probabilities of trigger for these higher threshold levels are extremely small and do not differ to a noticeable extent for small changes in the threshold level.

\subsection{Comparison of three conversion prices}
\label{sec:p3}
\noindent Finally, we compare 5-year IL CocoCat time-zero prices for three different conversion prices: when $K_P$ is a constant at a value of $8$ (denoted by $V_0^1$), $K_P = S_\tau$ (denoted by $V_0^2$) and  $K_P = S_\tau^{0.5}$ (denoted by $V_0^3$). Table \ref{Table_comparison_price} lists the time-zero prices of IL CocoCat for the different conversion prices, $V_0^1$, $V_0^2$ and $V_0^3$, for different threshold levels, $D$. We focus our comparison mainly on lower threshold levels since the differences in price, per threshold level, were visible to at least three significant figures. As we remarked in Section \ref{sec:p2}, for higher threshold levels price differences across different threshold levels were exceptionally small.

\begin{table}[H]
  \center
      \caption{Time-zero IL CocoCat prices for different threshold levels, across the three conversion price structures, $V_0^1$, $V_0^2$ and $V_0^3$, for a term of 5 years and $S_0 = 10$.}
  \label{Table_comparison_price}
    \begin{tabular}{l|c|c|c}
    \toprule
          \small{\textbf{$D$}} & \small{\textbf{$V_0^1$}} & \small{\textbf{$V_0^2$}} & \small{\textbf{$V_0^3$}} \\
    \midrule
\small{\textbf{$1.3 \times 10^{10} \qquad$}} & \small{\textbf{$\quad 0.345 \quad $}} & \small{\textbf{$\quad 0.310 \quad$}} & \small{\textbf{$\quad 0.331 \quad$}} \\
\small{\textbf{$1.8 \times 10^{10}$}} & \small{\textbf{$0.423$}} & \small{\textbf{$0.379$}} & \small{\textbf{$0.391$}} \\
\small{\textbf{$2.3 \times 10^{10}$}} & \small{\textbf{$0.523$}} & \small{\textbf{$0.460$}} & \small{\textbf{$0.487$}} \\
\small{\textbf{$2.9 \times 10^{10}$}} & \small{\textbf{$0.691$}} & \small{\textbf{$0.653$}} & \small{\textbf{$0.676$}} \\
\small{\textbf{$3.4 \times 10^{10}$}} & \small{\textbf{$0.952$}} & \small{\textbf{$0.915$}} & \small{\textbf{$0.948$}} \\
\small{\textbf{$4.0 \times 10^{10}$}} & \small{\textbf{$1.263$}} & \small{\textbf{$1.271$}} & \small{\textbf{$1.292$}} \\
\small{\textbf{$9.5 \times 10^{10}$}} & \small{\textbf{$1.507$}} & \small{\textbf{$1.508$}} & \small{\textbf{$1.509$}} \\
\small{\textbf{$2.5 \times 10^{11}$}} & \small{\textbf{$1.579$}} & \small{\textbf{$1.579$}} & \small{\textbf{$1.579$}} \\
\small{\textbf{$3.5 \times 10^{11}$}} & \small{\textbf{$1.579$}} & \small{\textbf{$1.579$}} & \small{\textbf{$1.579$}} \\
    \bottomrule
    \end{tabular}
\end{table}

For each of $V_0^1$, $V_0^2$ and $V_0^3$, it is still the case that the price levels out for very high threshold levels. This effect was also detected in Section \ref{sec:p1} and \ref{sec:p2}: the conversion price (and hence feature) has little effect on the IL CocoCat time-zero price. However, Table \ref{Table_comparison_price} shows that -- within the constrains of our simulation exercise -- $V_0^2$ is always less than or equal to $V_0^3$, which is expected since the conversion option is more valuable due to the conversion price being lower than the share price at time-of-conversion. However, it is interesting to analyse the relationship between $V_0^1$ and $V_0^2$. At lower threshold levels, it is clear from Table \ref{Table_comparison_price} that $V_0^1$ exceeds $V_0^2$. For these low threshold levels, it is likely that conversion will occur more quickly (than conversion for higher threshold levels), and moreover such small losses and the effect of changing interest-rates over a shorter term will not greatly impact the initial share price of $10$. Hence, at conversion time an investor in the IL CocoCat with a constant conversion price of $8$ will, in all likelihood, purchase the share for a value less than its current value at conversion time, leading to a higher value for the conversion feature. However, for higher threshold levels (i.e. $4.0 \times 10^{10}$ and higher), $V_0^2$ exceeds $V_0^1$.  At such high threshold levels, the large insured losses now have a more marked impact on the share price, leading to a potentially depressed share price at conversion time, which may well be lower than the fixed conversion price of $8$. Additionally it takes more time for the loss process to reach the higher threshold level which allows for further uncertainty in the interest-rate movements. Overall, this interest-rate uncertainty coupled with the possibility of depressed share prices will, in all likelihood, lead to a lower value for the conversion feature.

We end by emphasising that it is up to the issuer to select an appropriate threshold level for the contractually-specified conversion price. The threshold level should not be set too high (which is the option which will be preferred by the investor) so conversion never happens. But also, the conversion price should not be set too low, so that the probability of trigger is high (despite the lower IL CocoCat price), which is a situation that may not be favoured by the investor.

\section{Conclusions}
\label{section_conclusions}
\noindent In this research, we formalised the design of a CocoCat, which is a type of ILS similar to the traditional Coco bonds issued by banks. We linked it to the already existing Coco bond literature, and also briefly reviewed an existing CocoCat issued by SwissRe in 2013. Moreover, we motivated why there is a potential need for insurers and reinsurers to issue CocoCats. Amongst other reasons, CocoCats allow their issuers to tap into a broader pool of financing offered by the capital markets.

Subsequently, we went on to price a special type of Cococat, namely an IL CocoCat linked to the PCS loss index. The Longstaff model was chosen for the interest-rate dynamics. 
We were able to find intuitive, analytical expressions for the price via, firstly, our assumption of independence between catastrophe-risk and financial markets risk variables. Secondly, an exponential change of measure allowed us to separately deal with financial markets as well as catastrophe-risk variables, and a Girsanov-like transformation allowed us to synthetically remove a Brownian motion from the expectation containing two correlated Brownian motions. Finally, we arrived at an analytical expression for the conversion feature (and hence the price) which only required simulation of the loss process in order to empirically estimate the distribution of the time-of-trigger of the equity conversion feature. We did note that Monte Carlo simulation could be used to estimate the value of the conversion feature of the IL CocoCat directly. However, our simplification to an analytical formula had more in its favour, since only one process had to be simulated, namely the insured loss index, while the interest-rate and stock price processes did not have to be simulated.

Lastly, we presented a numerical analysis into the prices of the IL CocoCat, and we believe such an analysis is crucial in the design stage of the instrument and for use in pitch books. The prices we obtained in our analyses conformed to intuition: the higher the threshold level of the IL CocoCat, the greater the price. But for exceptionally high threshold levels, the IL CocoCat's price did not vary much which may be a limitation of the instrument. We also found evidence suggesting that the IL CocoCat price behaviour was quite sensitive to three design aspects: the threshold level, interest-rates and the conversion fraction.

Since the CocoCat has attracted little scholarly attention to date, further research into it is clearly called for. It would be interesting to look into other types of CocoCats, such as ones based on parametric triggers, and study their price behaviour. 
It would also be interesting to consider other design features such as converting the recovery upon trigger into the equity of an entity other than the issuer. 

\section*{Acknowledgments}
\noindent We are extremely grateful for the input of and insights offered by Peter Ouwehand as well as David Taylor (The African Institute for Financial Markets and Risk Management, University of Cape Town), Melusi Mavuso (Department of Statistical Sciences, University of Cape Town) and Jonas Becker (Willis Towers Watson, Munich). But most of all, we thank Mike Ward (GIBS, University of the Witwatersrand) for spurring on the idea.
This work was partially supported by Polish National Science Centre Grant  No.~2016/23/B/HS4/00566.

\vskip 0.5cm
\noindent{\bf References}\\
\setlength\bibsep{1pt}
\bibliographystyle{elsarticle-harv}
\bibliography{BibliographySample}

\end{document}